\documentclass[lettersize,journal]{IEEEtran}

\usepackage[utf8]{inputenc}  
\usepackage[T1]{fontenc}    
\usepackage{hyperref}       
\usepackage{url}          
\usepackage{booktabs}        
\usepackage{amsfonts}        
\usepackage{nicefrac}        
\usepackage{microtype}       
\usepackage{bm}
\usepackage{braket}
\usepackage{graphicx}
\usepackage{amsthm}
\usepackage{amssymb}
\usepackage{amsmath}
\usepackage[linesnumbered,ruled,vlined]{algorithm2e}

\usepackage{adjustbox}
\usepackage[normalem]{ulem}
\newtheorem{thm}{Theorem}
\newtheorem{lem}{Lemma}

\usepackage{mathtools}

\usepackage{bbm}
\usepackage{nccmath}
\usepackage[dvipsnames]{xcolor}
\usepackage{caption}
\usepackage{tkz-tab}
\usepackage{subcaption}
\usepackage{wrapfig}
\DeclareMathOperator{\Tr}{Tr}

\newcommand{\RY}{\mathop{\text{RY}}}
\newcommand{\RZ}{\mathop{\text{RZ}}}
\newcommand{\CNOT}{\mathop{\text{CNOT}}}

\DeclareMathOperator{\DQC}{\mathsf{DQC}1}
\DeclareMathOperator{\NP}{\mathsf{NP}}

\DeclareMathOperator{\hsigma}{\hat{\sigma}}

\DeclareMathOperator{\FF}{F}
\DeclareMathOperator{\hFide}{\widehat{F}}

\usepackage{cite}
 
\ifCLASSINFOpdf
\else
 
\fi

\begin{document}
 
\title{On exploring the potential of quantum auto-encoder for learning quantum systems}

\author{Yuxuan~Du 
        and~Dacheng~Tao,~\IEEEmembership{Fellow,~IEEE}

\IEEEcompsocitemizethanks{\IEEEcompsocthanksitem Y. Du and D. Tao are with College of Computing and Data Science, Nanyang Technological University, 639798, Singapore, and JD Explore Academy, Beijing, 101111, China.\protect\\
 
E-mail: duyuxuan123@gmail.com} 
}

\markboth{Journal of \LaTeX\ Class Files,~Vol.~14, No.~8, August~2021}
{Shell \MakeLowercase{\textit{et al.}}: A Sample Article Using IEEEtran.cls for IEEE Journals}

\maketitle

\begin{abstract}
The frequent interactions between quantum computing and machine learning revolutionize both  fields. One prototypical achievement  is the quantum auto-encoder (QAE), as the leading strategy to relieve the curse of dimensionality ubiquitous in the quantum world. Despite its attractive capabilities, practical applications of QAE have yet largely unexplored. To narrow this knowledge gap, here we devise three effective QAE-based learning protocols to address three classically computational hard learning problems when learning quantum systems, which are low-rank state fidelity estimation, quantum Fisher information estimation, and Gibbs state preparation. Attributed to the versatility of QAE, our proposals can be readily executed on near-term quantum machines. Besides, we analyze the error bounds of the trained protocols and showcase the necessary conditions to provide practical utility from the perspective of complexity theory. We conduct numerical simulations to confirm the effectiveness of the proposed three protocols. Our work sheds new light on developing advanced quantum learning algorithms to accomplish hard quantum physics and quantum information processing tasks.
\end{abstract}


\section{Introduction}\label{sec:introduction}
\IEEEPARstart{M}{anipulating} quantum systems and predicting their properties are of fundamental importance for developing quantum technologies \cite{biamonte2017quantum,carrasquilla2017machine,butler2018machine}. Despite the importance, exactly capturing all information of a given quantum system is computationally hard, since the space of a quantum state exponentially scales with the number of particles. For instance, quantum state tomography requests $\Theta(4^N)$ copies to determine an unknown $N$-qubit quantum state \cite{haah2017sample}; the recognition of entanglement of quantum states has proven to be NP-hard \cite{gurvits2003classical}. With the aim of understanding quantum world, two leading approaches, i.e., shadow-tomography-based methods \cite{aaronson2019gentle,aaronson2019shadow,huang2020predicting} and data compression methods \cite{gross2010quantum,rozema2014quantum,Shabani2011Est}, are developed to  extract useful information from investigated quantum systems from different angles. Namely, the former concentrates on predicting  certain properties of an unknown quantum state, e.g., estimating the expectation values of quantum observables. In other words, this approach aims to extract partial information from the explored quantum system. The latter delves into fully reconstructing a special class of quantum systems satisfying the low-rank property. Concisely, this approach first compresses the quantum state into a low-dimensional subspace, followed by quantum tomography with a polynomial runtime complexity.

 Among various quantum data compression techniques, the quantum auto-encoder (QAE) \cite{romero2017quantum} and its variants \cite{cao2021noise,lamata2018quantum}, as the quantum extension of classical auto-encoders \cite{goodfellow2016deep},   have recently attracted great attention. Conceptually, QAE harnesses a variational quantum circuit $U(\bm{\theta})$ to encode an initial input state into a low-dimensional latent space so that the input state can be recovered from this compressed state by a decoding operation $U(\bm{\theta})^{\dagger}$, as shown in Figure \ref{fig:QAE}. Analogous to their classical counterparts, QAE is flexible, which allows us to implement it on noisy intermediate-scale quantum (NISQ) machines \cite{preskill2018quantum} to solve different learning problems. In particular, QAE carried out on linear optical systems has been applied to reduce qutrits to qubits with low error levels \cite{pepper2019experimental} and to compress two-qubit states into one-qubit states \cite{huang2020realization}. Additionally, QAE has been realized on superconducting quantum processors to compress three-qubit quantum states \cite{lamata2018quantum}. Besides quantum data compression, initial studies have exhibited that QAEs can be exploited to denoise GHZ states which subject to spin-flip errors and random unitary noise \cite{bondarenko2020quantum} and to conduct efficient error-mitigation \cite{zhang2021generic} and data encoding \cite{bravo2021quantum}.  Despite the above achievements, practical applications of QAE in the field of quantum science remain largely unknown. 

To conquer aforementioned issues, here we revisit QAE from its spectral properties and explore its potentials in learning quantum systems \cite{gebhart2023learning,dunjko2022quantum}, as the prerequisite to use quantum techniques to benefit various fields such as spectroscopy, chemistry, and material sciences. This exploration also aligns with the recent perspective in the quantum computing community, suggesting that quantum machine learning has the potential to significantly accelerate quantum data analysis compared to classical methods \cite{cerezo2022challenges,huang2022quantum}. To do so, we exhibit that the compressed state of QAE living in the latent space conveys sufficient spectral information of the input $N$-qubit state $\rho$. This means that when $\rho$ is \textit{low-rank}, the spectral information of $\rho$ can be efficiently acquired via applying quantum tomography \cite{nielsen2010quantum} to the \textit{compressed state} in polynomial runtime with $N$. This result hints a number of applications of QAE in learning quantum systems, because numerous tasks in this field can be recast to capture the spectral information of the given state. Based on this understanding, we devise three QAE-based learning protocols, i.e., QAE-based fidelity estimator (Section \ref{sec:QAE-fide}), QAE-based quantum Fisher information estimator (Section \ref{sec:QFI-est}), and QAE-based Gibbs state solver (Section \ref{sec:QAE-gibbs}), to address three classically computational hard problems in learning quantum systems, i.e., low-rank state fidelity estimation \cite{miszczak2009sub},   quantum Fisher information estimation  \cite{petz2011introduction}, and   quantum Gibbs state preparation \cite{brandao2017quantum,poulin2009sampling}, respectively.  Our study highlights the importance of utilizing the idea generated from the classical machine learning community to tackle hard problems in quantum science, which could have an impact on both computer science and quantum physics.

We summarize our main contributions  as follows.
 
\begin{figure*} 
	\centering
\includegraphics[width=0.76\textwidth]{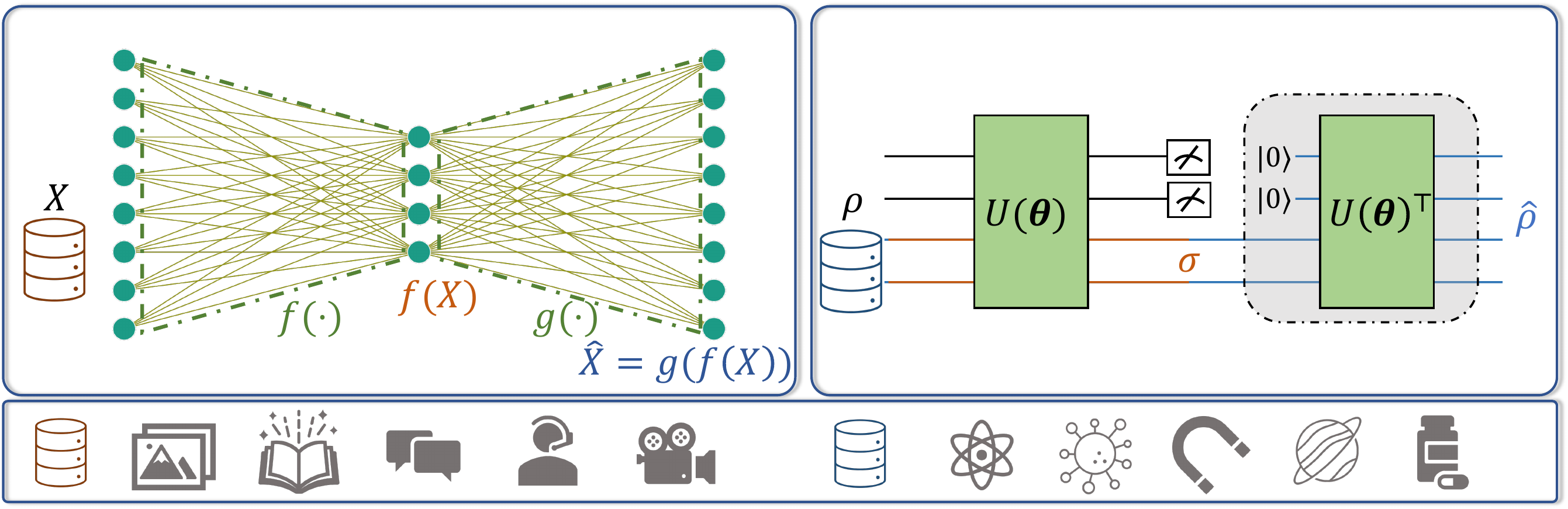}
\caption{\small{\textbf{The paradigm of AE and QAE.} 
The upper left panel illustrates the classical auto-encoder (AE). The upper right panel shows the quantum AE (QAE). In the training process, all quantum operations in the grey box are removed. The state $\sigma$ corresponds to the compressed state of $\rho$ (see Section \ref{subsec:QAE} for explanations). The lower panel depicts the data type suited for AE and QAE. Two typical datasets explored by QAE are collected from chemistry and material science.
}}
\label{fig:QAE}  
\end{figure*}

\begin{itemize}
\item We present three specific learning protocols that utilize QAE to accelerate quantum data analysis in low-rank state fidelity estimation, quantum Fisher information estimation, and quantum Gibbs state preparation. These protocols hold promise for applications in quantum materials, biochemistry, and high-energy physics.
\item We establish theoretical foundations of the proposed three protocols. In particular, for QAE-based fidelity estimator, we prove that the estimation error  of our proposal is upper bounded by $\sqrt{2\varsigma}$, where $\varsigma$ is a computationally-efficient quantity measuring the discrepancy between the target state and the reconstructed state of QAE after $T$ iterations. For QAE-based quantum Fisher information estimator, we prove that its estimation error is upper bounded by $\sqrt{2\varsigma}/\tau^2$, where $\tau$ is the shifted parameter of the source being probed.  For QAE-based Gibbs state solver, we prove that its estimation error, measured by the discrepancy between the produced state and target state, yields $O(\sqrt{\beta + 2\varsigma\log\frac{2}{1 - \sqrt{1 - \delta}}})$, where $\beta$ is the inverse temperature and $\delta$ is the training loss of QAE.  
\item We present the necessary conditions when using our protocols to gain potential computational merits in the view of complexity theory. For example, in fidelity estimation,   constraining the number of iterations of training QAE as $T \sim  O(poly(N))$, our proposal can surpass classical methods if $\varsigma < \text{poly}(1/N)$, supported by the results that fidelity estimation of low-rank states is $\DQC$-hard \cite{cerezo2020variational}.   
\item  We conduct extensive numerical simulations up to $8$ qubits to demonstrate the effectiveness of our proposals.   
\end{itemize}

\medskip
The organization of this paper is as follows. In Section \ref{sec:backII}, we present backgrounds of quantum computing, recap the mechanism of AE and QAE, and introduce the related work. Then, we analyze the spectral property of QAE in Section \ref{sec:specIII}. By leveraging the obtained spectral property, we propose three QAE-based learning protocols, i.e., QAE-based fidelity estimator,  QAE-based quantum Fisher information estimator, and QAE-based  quantum Gibbs state solver in Sections \ref{sec:QAE-fide}, \ref{sec:QFI-est}, and \ref{sec:QAE-gibbs}, respectively. In Section \ref{sec:prac_implement}, we discuss how to realize our proposals on NISQ machines. Last, in Section \ref{sec:conclusion}, we conclude this study and discuss future research directions.    
 
\section{Background}\label{sec:backII}
\subsection{Basic concepts of quantum computing}

The fundamental unit in quantum computing is qubit, which refers to a two-dimensional vector. Under Dirac notation, a qubit state is denoted by  $\ket{\bm{\alpha}}=a_0\ket{0}+a_1\ket{1}\in\mathbb{C}^2$, where $\ket{0}=[1,0]^{\top}$ and $\ket{1}=[0,1]^{\top}$ specify two unit bases, and the coefficients $a_0, a_1 \in \mathbb{C}$ satisfy $|a_0|^2+|a_1|^2=1$. An $N$-qubit state is denoted by $ \ket{\Psi}=\sum_{i=1}^{2^N}a_i\ket{i} \in \mathbb{C}^{2^N}$, where $\ket{i}\in\mathbb{R}^{2^N}$ is the unit vector whose $i$-th entry being 1 and other entries are 0, and $\sum_{i=0}^{2^N-1}|a_i|^2=1$ with $a_i\in\mathbb{C}$. Following conventions, an $N$-qubit separable state can  be written as  $\ket{\bm{a}_1,...\bm{a}_N}$ or $\ket{\bm{a}_1}...\ket{\bm{a}_N}$. We sometimes denote $\ket{\bm{a}}\ket{\bm{b}}$ as $\ket{\bm{a}}_A\ket{\bm{b}}_B$, which means that the qubits $\ket{\bm{a}}_A$ ($\ket{\bm{b}}_B$) are assigned in the quantum register $A$ ($B$). Besides Dirac notation, the density matrix can be used to describe more general qubit states. For example, the density matrix of the state $\ket{\Psi}$ is $\rho=\ket{\Psi}\bra{\Psi}\in \mathbb{C}^{2^N\times 2^N}$. For a set of qubit states $\{p_i, \ket{\psi_i} \}_{i=1}^m$ with $p_i>0$, $\sum_{i=1}^{m}p_i=1$, and $\ket{\psi_i}\in \mathbb{C}^{2^N}$ for $\forall i\in [m]$, its density matrix is $\rho=\sum_{i=1}^{m} p_i\rho_i$ with $\rho_i=\ket{\psi_i}\bra{\psi_i}$ and $\Tr({\rho})=1$.
   
Quantum gates refer to unitary transformations and serve as the computational toolkit for quantum circuit models, i.e., an $N$-qubit gate $U\in\mathcal{U}(2^N)$ obeys $UU^{\dagger}=U^{\dagger}U=\mathbb{I}_{2^N}$, where $\mathcal{U}(\cdot)$ stands for the unitary group. Throughout the whole study, we focus on the single-qubit and two-qubit quantum gate set, i.e., $\RZ(\theta) = \begin{psmallmatrix} e^{-i\theta/2} & 0\\ 0 & e^{i\theta/2}\end{psmallmatrix}$,  $\RY(\theta) = \begin{psmallmatrix} \cos{\theta/2} & -\sin{\theta/2}\\ \sin{\theta/2} & \cos{\theta/2} \end{psmallmatrix}$, and $\CNOT=\begin{psmallmatrix} 1 & 0 & 0 & 0 \\ 0 & 1 & 0 & 0 \\ 0 & 0 & 0 & 1 \\  0 & 0 & 1 & 0 \end{psmallmatrix}$.  Note that this gate set is universal such that these quantum gates can be used to reproduce the functions of all the other quantum gates \cite{nielsen2010quantum}.  

Quantum measurements are employed to extract quantum information of the evolved quantum state into the classical form. In this study, we concentrate on the positive operator-valued measures (POVM), which is described by a collection of positive operators $0\preceq \Pi_i$ satisfying $\sum_i\Pi_i=\mathbb{I}$.  Specifically, applying the measurement $\{\Pi_i\}$ to the state $\rho$, the probability of outcome $i$ is given by $\Pr(i) = \Tr(\rho\Pi_i)$. 

 \subsection{Quantum machine learning}
Quantum machine learning (QML) employs quantum computers to gain computational advantages in learning both classical data (e.g., images, videos, and texts) and quantum data (e.g., dynamics and correlation properties of quantum systems) \cite{biamonte2017quantum}. Learning quantum systems is significant as quantum technology development relies on creating and manipulating quantum systems \cite{huang2021information,huang2022provably}, which become exponentially complex with size, making their full description intractable \cite{gebhart2023learning,dunjko2022quantum}. Thus, a key research goal is to develop QML models that surpass classical models in solving crucial problems in quantum science such as quantum computing, communication, and metrology \cite{pezze2018quantum, degen2017quantum}.

A leading candidate of QML models in NISQ era is  variational quantum algorithms (VQAs), also known as quantum neural networks (QNNs) \cite{cerezo2020variational2,beer2020training,farhi2018classification,benedetti2019parameterized,du2018expressive,mitarai2018quantum,Philip2022Efficient,Carvalho2023Parameterized}.  Conceptually, a VQA contains four principal blocks: the input state $\rho_0$, a variational quantum circuit (or equivalently ansatz) $U(\bm{\theta})$,  the measurement $M$,  and a loss function $\mathcal{L}(\bm{\theta})$. In the training stage, VQA follows an iterative manner to proceed optimization, where the classical optimizer continuously leverages the output of the quantum circuit, i.e., $\Tr(MU(\bm{\theta})\rho_0U(\bm{\theta})^{\dagger})$, to update trainable parameters $\bm{\theta}$ to minimize $\mathcal{L}(\bm{\theta})$. The updating rule at the $t$-th iteration   is $\bm{\theta}^{(t+1)}=\bm{\theta}^{(t)} - \eta \nabla_{\bm{\theta}} \mathcal{L}(\bm{\theta})$, where $\eta$ is the learning rate and $\nabla_{\bm{\theta}} \mathcal{L}(\bm{\theta})$ is the gradients of $\mathcal{L}(\bm{\theta})$ that can be acquired via parameter shift rule \cite{schuld2019evaluating}. The updating is terminated if $\mathcal{L}(\bm{\theta})$ is converged or the  iteration $t$ exceeds the threshold.

 \subsection{Classical auto-encoder}\label{sebsec:clc-ae} 
Auto-encoder (AE) is a technique in the context of machine learning  to learn efficient codings of unlabeled data and has been broadly applied to clustering \cite{guo2019adaptive}, spatio-temporal data mining \cite{wang2020deep}, information retrieval \cite{salakhutdinov2009semantic}, anomaly detection \cite{zhou2017anomaly}, and image processing \cite{vincent2008extracting}. Intuitively, AE completes the data compression by mapping the given input from a high-dimensional space to a low-dimensional space in the sense that the input can be recovered from this dimension-reduced point. As shown in the upper left panel of Figure \ref{fig:QAE}, two basic elements of AE are the encoder $f(\cdot)$ and the decoder $g(\cdot)$, which are realized by neural networks \cite{goodfellow2016deep}. Suppose that the given dataset is $X\in\mathbb{R}^{m\times n}$ whose $i$-th column $\bm{x}_i\in \mathbb{R}^m$ represents the $i$-th data point. The encoder $f:\mathbb{R}^m\rightarrow \mathbb{R}^k$ with $k\ll m$ compresses each data point into a $k$-dimensional space $f(\bm{x}_i)$, while the decoder $g:\mathbb{R}^k \rightarrow \mathbb{R}^m$ maps the latent representation $f(\bm{x}_i)$ to the reconstructed data $\hat{\bm{x}}_i=g(f(\bm{x}_i))$ \cite{goodfellow2016deep}.   The optimization of AE amounts to minimizing the reconstruction error between   $X$ and its reconstruction $\hat{X}$ whose $i$-th column is $\hat{\bm{x}}_i$.

\subsection{Quantum auto-encoder}\label{subsec:QAE} 
Quantum auto-encoder (QAE), as the quantum analog of AE, is composed of the quantum encoder and the quantum decoder as exhibited in the upper right panel of Figure \ref{fig:QAE}. The construction of the quantum encoder (or decoder) is completed by variational quantum circuits. Mathematically, define an $N$-qubit input state as $\rho\in\mathbb{C}^{2^N\times 2^N}$ and the number of qubits that represents the latent space as $K$ with $K<N$. Let the quantum encoder be a variational quantum circuit $U(\bm{\theta})$, where $U(\bm{\theta})=\prod_{l=1}^LU_l(\bm{\theta})$ and $U_l(\bm{\theta})$ is a sequence of parameterized single-qubit and two-qubits gates with $\bm{\theta}$ being tunable parameters.  Under the above setting, QAE aims to find optimal parameters $\bm{\theta}^*$ that minimize the loss function   
 \begin{eqnarray}\label{eqn:loss_G1}
	&& \mathcal{L}(U(\bm{\theta}),\rho) = \Tr\left(M_GU(\bm{\theta})\rho U(\bm{\theta})^{\dagger} \right),
\end{eqnarray}
where $M_G = \left(\mathbb{I}_N-(\ket{0}\bra{0})^{\otimes(N-K)}\otimes \mathbb{I}_K \right)$ refers to the measurement operator and $\mathbb{I}_N\in\mathbb{R}^{2^N\times 2^N}$ is the identity.  While both AEs and QAEs use encoders and decoders to process data, QAEs are  specifically designed for quantum computers and are therefore better suited for learning quantum data, as illustrated in the lower panel of Figure \ref{fig:QAE}, in contrast to AEs which are typically used for classical data.

\subsection{Related work}
 Previous literature related to our work can be divided into two categories. The first category is QAEs and the second group is other quantum learning methods to accomplish the tasks explored in this study. Here we separately compare them with our work.

The variants of QAEs have been proposed to compensate certain deficiencies existed in the original QAE proposed by \cite{romero2017quantum}. In particular, the proposal developed in \cite{lamata2018quantum,ding2019experimental} combines QAE with quantum adders, which can approximately add two unknown quantum states supported in different quantum systems. Ref.~\cite{srikumar2021clustering} proposed a hybrid QAE to tackle clustering and classification task. On par with the model design, initial studies have been conducted to explore applications of QAEs in state compression \cite{pepper2019experimental,huang2020realization}, denoising  \cite{bondarenko2020quantum} and error-mitigation \cite{zhang2021generic}. Different from prior studies that mainly focus on the feasibility of QAEs, our work explores their practical applications in learning quantum systems.

Several studies have designed other variational quantum algorithms to  address the learning problems exploited in this study. For the task of fidelity estimation, the protocol proposed by Ref.~\cite{chen2021variational} can be applied to full-rank quantum states; the protocol proposed by Ref.~\cite{cerezo2020variational} intended to the low-rank fidelity estimation. Compared with these two method, our proposal is more  resource-efficient.  For the task of Gibbs state preparation, a relevant study is  Ref.~\cite{wang2020variational}. In that study, the protocol leverages the technique of Taylor truncation to estimate Von-Neumann entropy, which introduces the estimation error. Conversely, our approach does not require such an estimation.

\section{Spectral properties of QAE}\label{sec:specIII} 
We now revisit QAE from its spectral properties. Denote the rank of the  state $\rho$ as $r^*$ and its spectral decomposition as  
\begin{equation}\label{eqn:eigen-rho}
	\rho=\sum_{i=1}^{r^*}\lambda_i \ket{\psi_i}\bra{\psi_i}~\text{with}~\sum_{i=1}^{r^*}\lambda_i =1,
\end{equation}
where $\lambda_i$ and $\ket{\psi_i}$ are the $i$-th eigenvalues and eigenvectors, respectively. Without loss of generality, we suppose  $\lambda_1\geq \lambda_2 \geq ... \geq \lambda_{r^*}$. The optimal $U(\bm{\theta}^*)$ in Eqn.~(\ref{eqn:loss_G1}) is shown in the following theorem whose proof is given in Appendix \ref{Appendix:thm1_opt_res}. 
\begin{thm}\label{thm_QAE_noiseless}
	QAE with the loss function $\mathcal{L}(U(\bm{\theta}),\rho)$ in Eqn.~(\ref{eqn:loss_G1}) has multiple critical (global minimum) points of $U(\bm{\theta}^*)$. Let $r^*=2^{K}$. The generic form of $U(\bm{\theta}^*)$ is
	\begin{equation}\label{eqn:thm1_0}
		U(\bm{\theta}^*) := \sum_{i=1}^{r^*} \ket{0}^{\otimes N-K}\ket{\varpi_i}\bra{\psi_i} + \sum_{j=r^*+1}^{2^N}\ket{\phi_j}\bra{\psi_j},
	\end{equation}
	where $\{ \ket{0}^{\otimes N-K}\ket{\varpi_i}\}_{i=1}^{r^*}\cup \{\ket{\phi_j}\}_{j=r^*+1}^{2^N} $ and $\{\ket{\psi_j}\}_{j=1}^{2^N}$ are two sets of orthonormal vectors and $\{\ket{\psi_i}\}_{i=1}^{r^*}$ is defined in Eqn.~(\ref{eqn:eigen-rho}). 
\end{thm}

Given access to the optimal quantum encoder $U(\bm{\theta}^*)$ in Eqn.~(\ref{eqn:thm1_0}), the compressed state of QAE yields  
\begin{equation}
		\sigma^* = \Tr_{E}\left(\frac{M_G U(\bm{\theta}^*)\rho U(\bm{\theta}^*)^{\dagger} M_G^{\dagger}}{\Tr(M_G U(\bm{\theta}^*)\rho U(\bm{\theta}^*)^{\dagger})} \right), 
\end{equation} 
where the term `$\Tr_E(\cdot)$' represents the partial trace of all qubits except for the $K$ latent qubits \cite{nielsen2010quantum}. The spectral information between $\sigma^*$ and $\rho$ is closely related, as indicated by the following lemma, where the corresponding proof is provided  in Appendix \ref{append:proof_lem1}.
\begin{lem}\label{lem:sigma_eigen}
	The compressed state has the same eigenvalues with $\rho$, i.e., $\sigma^* = \sum_{i=1}^{r^*}\lambda_i\ket{\varpi_i}\bra{\varpi_i}$. Moreover, the subspace spanned by the eigenvectors $\{\ket{\psi_i}\}_{i=1}^{r^*}$ of $\rho$ can be recovered by $U(\bm{\theta}^*)$ and eigenvectors $\{\ket{\varpi}_i\}_{i=1}^{r^*}$.  
\end{lem}

The result of Lemma \ref{lem:sigma_eigen} indicates that when the input state $\rho$ is low-rank, the spectral information of $\sigma^*$ can be efficiently obtained by quantum tomography associated with classical post-processing. In particular, the size of the $K$-qubit compressed state $\sigma^*$ scales with $r^*\times r^*$, which implies that the classical description of $\sigma^*$ can be obtained in $O(\text{poly}(r^*))$ runtime. Moreover, the collected classical information suggests that the spectral decomposition of $\sigma^*$, i.e., $\{\lambda_i\}$ and $\{\ket{\varpi_i}\}$, can be computed in $O(\text{poly}(r^*))$ runtime. Taken together, the spectral information of a low-rank state $\rho$ can be efficiently acquired via QAE in the optimal case.

\begin{figure*}
	\centering
\includegraphics[width=0.99\textwidth]{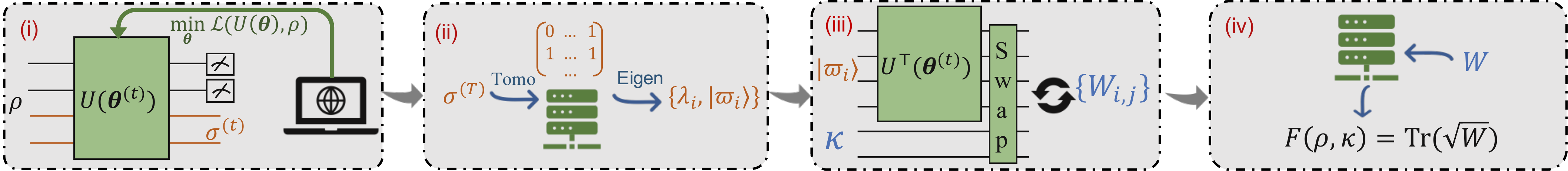}
	\caption{\small{\textbf{The paradigm of the QAE-based fidelity estimator.} The implementation of the QAE-based fidelity estimator is composed of four steps. (i) QAE in Eqn.~(\ref{eqn:loss_G1}) is applied to optimize $\bm{\theta}$ with $T$ iterations. (ii) After optimization, the compressed state $\sigma^{(T)}$ in Eqn.~(\ref{eqn:sigma_T_1}) is extracted into the classical register via quantum state tomography. Once the classical form of $\sigma^{(T)}$ is accessible, the classical Eigen-solver is employed to acquire its spectral information, i.e., $\{\lambda_i, \ket{\varpi_i}\}$. (iii) When the spectral information of $\rho$ is available and the state $\kappa$ is accessible, we can calculate $W$ in Eqn.~(\ref{eqn:QAE-W}) based on Ref.~\cite{cerezo2020variational}. (iv) Given access to $W$, the fidelity $F(\rho,\kappa)$ can be effectively obtained. }}
	\label{fig:QAE-fide1}
\end{figure*}

 In the subsequent sections, we elucidate how to employ QAE, especially the acquired spectral information, to facilitate the low-rank state fidelity estimation, quantum Fisher information estimation, and Gibbs state preparation tasks. Due to the intrinsic hardness, any such improvement would not only be immensely valuable for various downstream tasks in quantum science such as quantum computing, quantum sensing, and quantum metrology but also indicates the practical utility of near-term quantum computers. For clarity, we unify some notations throughout the rest study. Denote the spectral decomposition of the variational quantum circuits $U(\hat{\bm{\theta}})$ as
\begin{equation}\label{eqn:thm_est_fide_1}
	U(\hat{\bm{\theta}})  := \sum_{i=1}^r  \ket{0}^{\otimes N-K}\ket{\hat{\varpi}_i}\bra{\hat{\varphi}_i} + \sum_{j=r+1}^{2^N} \ket{\hat{\phi}_j}\bra{\hat{\varphi}_j},
\end{equation}	
where $\hat{\bm{\theta}}$ refers to  the trained parameters, and $\{\ket{0}^{\otimes N-K}\ket{\hat{\varpi}_i}\}_{i=1}^r \bigcup \{\ket{\hat{\phi}_j}\}_{j=r+1}^{2^N}$ and $\{\ket{\hat{\varphi}_i}\}_{i=1}^{2^N}$ are two sets of orthonormal vectors with $r=2^K$. Note that $r$ may not be equal to $r^*$ when the training procedure is not optimal. Besides, since the Eigen-decomposition of $U(\hat{\bm{\theta}})$ is not unique and  it operates with the projector $M_G$ in Eqn.~(\ref{eqn:loss_G1}), we set $U(\hat{\bm{\theta}})$  a similar form as shown in Eqn.~(\ref{eqn:thm1_0}) for ease of analysis. Based on Eqn.~(\ref{eqn:thm_est_fide_1}), the explicit form of the compressed state of QAE in the non-optimal case is 
\begin{equation}\label{eqn:thm2_sigma}
	\hsigma = \Tr_{E}\left(\frac{M_G U(\hat{\bm{\theta}})\rho U(\hat{\bm{\theta}})^{\dagger} M_G^{\dagger}}{\Tr(M_G U(\hat{\bm{\theta}})\rho U(\hat{\bm{\theta}})^{\dagger})} \right) = \sum_{i=1}^r\hat{\lambda}_i \ket{\hat{\varpi}_i}\bra{\hat{\varpi}_i},
\end{equation}
where $\hat{\lambda}_i={\braket{\hat{\varphi}_i|\rho|\hat{\varphi}_i}}/\left({\sum_{j=1}^r \braket{\hat{\varphi}_j|\rho|\hat{\varphi}_j}}\right)$.

 \section{QAE-based fidelity estimator}\label{sec:QAE-fide} 
The quantification of the quantum state fidelity plays a central role in verifying and characterizing the states prepared by a quantum processor \cite{nielsen2010quantum} and studying quantum phase transitions \cite{gu2010fidelity}.  In other words, devising an effective fidelity estimator  is the precondition of using quantum computer to benefit chemistry, quantum physics, and material sciences.   Despite its importance, quantum state fidelity calculation is computationally hard for both classical and quantum devices \cite{cerezo2020variational}. Mathematically, the fidelity of two states $\rho$ and $\kappa$  yields
\begin{equation}\label{eqn:fide}
	\FF(\rho,\kappa) = \Tr\left(\sqrt{\sqrt{\rho}\kappa\sqrt{\rho}}\right)=\|\sqrt{\rho}\sqrt{\kappa}\|_1.
\end{equation}
  However, when we slightly relax the problem and target to the low-rank state fidelity estimation (i.e., the rank of $\rho$ is no larger than $O(poly(N))$ and $\kappa$ is unconstrained), quantum computers may provide potential advantages over classical devices, supported by the following lemma. 
\begin{lem}[Proposition 5, \cite{cerezo2020variational}]\label{lem:DQC-1}
The problem of the $N$-qubit low-rank state fidelity estimation within the precision $\pm 1/\text{poly(N)}$ is $\DQC$-hard. 
\end{lem}
\noindent Recall that classical computers can not efficiently simulate $\DQC$ models unless the polynomial-time hierarchy collapses to the second level \cite{fujii2018impossibility}. The result of Lemma \ref{lem:DQC-1} shows that if a quantum algorithm can achieve low-rank state fidelity estimation in \textit{polynomial runtime} with an error below $\text{poly}(1/N)$, it has \textit{a runtime advantage} compared to the optimal classical method.

Here we devise a QAE-based fidelity estimator to complete the low-rank state fidelity estimation task and discuss its potentials. The underlying concept of our proposal is as follows. Recall an equivalent formula of  fidelity is
\begin{equation}\label{eqn:QAE-W}
\FF(\rho,\kappa) = \Tr\Big(\sqrt{ \sum_{i,j} W_{i,j}\ket{\psi_i}\bra{\psi_j} }\Big),	
\end{equation}
where $ W_{i,j} = \sqrt{\lambda_i\lambda_j} \braket{\psi_i|\kappa|\psi_j}$ for $\forall i\in [r^*]$. In other words, calculating $\FF(\rho,\kappa)$ amounts to acquiring $W\in \mathbb{C}^{r^*\times r^*}$, which can be effectively achieved once the state $\kappa$ and the eigenvalues and eigenvectors of $\rho$ are accessible.   The spectral information of $\rho$ provided by QAE allows us to  estimate the fidelity $\FF(\rho,\kappa)$. 

\subsection{Implementation of QAE-based fidelity estimator}\label{sec:implementation-QAE-FE}
 We now elaborate on the protocol of the QAE-based fidelity estimator. As shown  in Figure \ref{fig:QAE-fide1}, the proposed method consists of four steps. 
\begin{enumerate}
	\item We apply QAE formulated in Eqn.~(\ref{eqn:loss_G1}) to compress the given state $\rho$. Let the total number of training iterations be $T$. The updating rule of the parameters $\bm{\theta}$ at the $t$-th iteration yields
\begin{equation}\label{eqn:QAE-update}
	\bm{\theta}^{(t+1)} = 	\bm{\theta}^{(t)} - \eta \frac{\partial \mathcal{L}(U(\bm{\theta}^{(t)}), \rho)}{\partial \bm{\theta}^{(t)}},
\end{equation}
where $\eta$ refers to the learning rate. The gradients $\partial \mathcal{L}(U(\bm{\theta}^{(t)}), \rho)/\partial \bm{\theta}^{(t)}$ can be acquired by using the parameter shift rule \cite{mitarai2018quantum,schuld2019evaluating}. After the optimization is completed, the trained parameters $\bm{\theta}^{(T)}$ are exploited to prepare the compressed state in Eqn.~(\ref{eqn:thm2_sigma}),  i.e.,
\begin{equation}\label{eqn:sigma_T_1}
		\sigma^{(T)} = \Tr_{E}\left(\frac{M_G U(\bm{\theta}^{(T)})\rho U(\bm{\theta}^{(T)})^{\dagger} M_G^{\dagger}}{\Tr(M_G U(\bm{\theta}^{(T)} )\rho U(\bm{\theta}^{(T)})^{\dagger})} \right). 
\end{equation} 
\item The quantum state tomography techniques \cite{nielsen2010quantum} are applied to $\sigma^{(T)}$ with  size $2^K\times 2^K$ to obtain its classical description. For simplicity, here we only focus on the ideal case with the infinite number of measurements so that the state $\sigma^{(T)}$ can be exactly recorded in the classical register. When the measurement error is considered, the likelihood recovering strategy can be used to efficiently find the closest probability distribution of $\sigma^{(T)}$ \cite{smolin2012efficient}. Once the classical form of $\sigma^{(T)}$ is in hand, its spectral information can be effectively obtained when the number of latent qubits $K$ is small. With a slight abuse of notation, the eigenvalues and eigenvectors of $\sigma^{(T)}$ are denoted by $\{\hat{\lambda}_i\}_{i=1}^{r}$ and $\{\hat{\varpi}_i\}_{i=1}^{r}$ with $r=2^K$, respectively, i.e.,
\begin{equation}
	\sigma^{(T)} = \sum_{i=1}^r\hat{\lambda}_i \ket{\hat{\varpi}_i}\bra{\hat{\varpi}_i}.
\end{equation}
Given access to $\{\ket{\hat{\varpi}_i}\}$ and $U(\bm{\theta}^{(T)})$, we can prepare quantum states $\{\ket{\hat{\varphi}_i}\}$, i.e.,
\begin{equation}\label{eqn:res-state-QAE}
	\ket{\hat{\varphi}_i} = U(\bm{\theta}^{(T)})^{\dagger}\ket{0}^{\otimes (N-k)}\ket{\hat{\varpi}_i}. 
\end{equation} 
 
\item The classical spectral information of $\sigma^{(T)}$, the variational quantum circuits $U(\bm{\theta}^{(T)})^{\dagger}$, and the state $\kappa$ are employed to estimate $W_{ij}$ in Eqn.~(\ref{eqn:QAE-W}). The element $W_{ij}$ is approximated by  
\begin{equation}\label{eqn:est-W}
	 \widehat{W}_{ij}= \sqrt{\hat{\lambda}_i\hat{\lambda}_j}\langle \hat{\varphi}_i|\kappa| \hat{\varphi}_j \rangle,~\forall i, j\in[r].
\end{equation}
Possible physical implementations to acquire $\langle \hat{\varphi}_i|\kappa| \hat{\varphi}_j \rangle$ are Swap test circuit and the destructive Swap test circuit \cite{subacsi2019entanglement}, which is identical to \cite{cerezo2020variational}.  
\item  The fidelity $F(\rho,\kappa)$ is approximated via the classical post-processing. In particular, given access to the matrix $ \widehat{W}$ with size $2^K\times 2^K$, we use $\hFide(\rho,\kappa)= \Tr(  \widehat{W}^{-1/2})$ to estimate $\FF(\rho,\kappa)$. 
\end{enumerate}
Seeking the optimal parameter $\bm{\theta}^*$ of QAE with the zero training loss is challenging with a finite number of iterations. Suppose the training loss of QAE after $T$ iterations as $\mathcal{L}(U(\bm{\theta}^{(T)}),\rho)=\delta$.  Denote $r=2^K$ and the compressed state of $\rho$ as $\sigma^{(T)}\in\mathbb{C}^{r\times r}$ and its spectral decomposition as $\{\hat{\lambda}_i\}$ and $\{\ket{\hat{\varpi}_i}\}$.  Given the spectral information of $\sigma^{(T)}$, the fidelity $\FF(\rho,\kappa)$ is estimated by  
\begin{equation}\label{eqn:est-fide}
	\hFide(\rho,\kappa) = \Tr((\sum_{i,j}  \sqrt{ \hat{\lambda}_i\hat{\lambda}_j} \braket{\hat{\varphi}_i|\kappa|\hat{\varphi}_i} \ket{\hat{\varphi}_i}\bra{\hat{\varphi}_j})^{-1/2}), 
\end{equation}
where $\ket{\hat{\varphi}_i} = U( \bm{\theta}^{(T)})\ket{0^{\otimes N-K},\varpi_i}$. An important question is quantifying the estimation error, i.e., $|\hFide(\rho,\kappa)-\FF(\rho,\kappa)|$. The following theorem indicates that the estimation error can be bounded by the quantity 
\begin{equation}\label{eqn:est-fide-error-quantity}
\varsigma = \sqrt{1- \FF^2(\rho, \hat{\rho})},
\end{equation}
where $\FF(\hat{\rho}, \rho)\equiv \hFide(\rho, \rho)$ refers to the fidelity between the reconstructed state of the trained QAE $\hat{\rho}$ and the target state $\rho$. The evaluation of $\varsigma$ is computationally efficient, since as the calculation of $\hFide(\rho, \kappa)$, $\FF(\hat{\rho}, \rho)$ can be efficiently attained by using Eqn.~(\ref{eqn:est-fide}).

\begin{thm}\label{thm:fide_bound}
Following notations in Eqns.~(\ref{eqn:est-fide}) and (\ref{eqn:est-fide-error-quantity}), the QAE-based fidelity estimator after $T$ iterations returns the estimated fidelity $\hFide(\rho,\kappa)$  with 
 \begin{equation}\label{eqn:thm2}
 	\hFide(\rho,\kappa) - \sqrt{2\varsigma} \leq  \FF(\rho, \kappa)\leq  \hFide(\rho,\kappa) + \sqrt{2\varsigma}.
 \end{equation}
\end{thm} 
\noindent The proof of Theorem \ref{thm:fide_bound} is provided in Appendix \ref{append:thm2}. The achieved results indicate how the computationally-efficient quantity $\varsigma$ reflects the upper bound of the estimation error $\hFide(\rho,\kappa)-\FF(\rho,\kappa)$ when the optimization of QAE is \textit{not optimal}. Moreover, the estimation error in Eqn.~(\ref{eqn:thm2}) hints  conditional advantages of our proposal. In particular, after a polynomial number of iterations with $T \sim  O(poly(N))$, our proposal could reach runtime advantages when   $\sqrt{2\varsigma}\leq 1/poly(N)$. This is because such a restriction on $T$ leads that the total runtime complexity of our proposal polynomially scales with $N$ to attain a low estimation error, which is unachievable by classical methods, according to Lemma \ref{lem:DQC-1}.  

Another intriguing property of the quantity $\varsigma$ is that it is lower bounded by the training loss of QAE $\delta$, as elaborated on Appendix~\ref{append:train-loss-vs-error-quantity}. In this regard, to suppress the estimation error of QAE-based fidelity estimator, the key  is suppressing $\delta$. Possible strategies towards this goal encompass the employment of problem-oriented ansatz to implement the quantum encoder $U(\bm{\theta})$ \cite{du2020quantumQAS} and the adoption of advanced optimization techniques \cite{cerezo2020cost,zhang2020toward} to improve the trainability of QAE. Moreover, it is noteworthy that according to the computational complexity theory, there may exist a class of computationally-hard quantum states can not be efficiently compressed by QAE in which $ \varsigma >  1/\text{poly}(N)$.

\begin{figure*} 
	\centering
\includegraphics[width=0.948\textwidth]{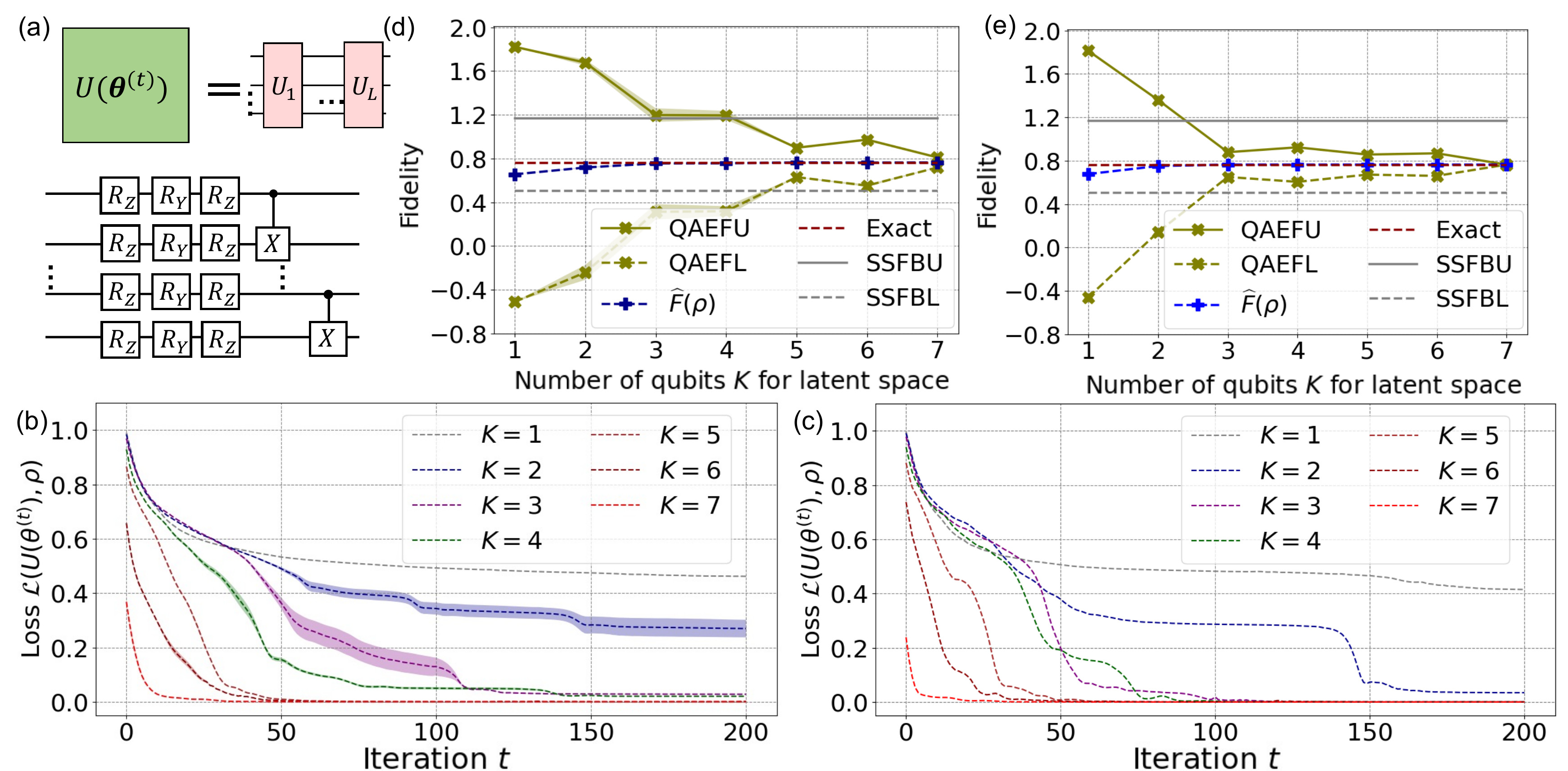}
\caption{\small{\textbf{Simulation results for low-rank states fidelity estimation with the number of qubits $N=8$}. (a) The implementation of variational circuit $U(\bm{\theta})$ used in the quantum encoder. The trainable parameters are contained in the rotational single-qubit gates RZ and RY. The entangled layer is composed of CNOT gates. (b) The average training loss of the QAE-based fidelity estimator in the first stage. The label `$K = a$' refers to the employed number of latent qubits of the QAE-based fidelity estimator is $a$. The shadow region refers to the variance. (c) The top-1 training loss achieved by QAE-based fidelity estimator. All notations follow the same meaning as those in (b). (d) The average fidelity estimation results returned by QAE-based fidelity estimator. The labels QAEFL and QAEFU refer to the bound in Theorem 2, respectively. The label `SSFBU' (`SSFBL') is super-fidelity (sub-fidelity) bounds.  (e) The top-1 fidelity estimation results achieved by QAE-based fidelity estimator. All notations follow the same meaning as those in (d). }}
	\label{fig:QAE-fide-low-append}
\end{figure*}

\noindent{\textbf{Remark}}.  We emphasize two key points regarding the QAE-based fidelity estimator. First, our analysis aligns with convention of NISQ algorithms to highlight the conditional advantages of our approach, contingent upon the fulfillment of specific requirements \cite{cerezo2022challenges,schuld2021supervised}. This contrasts from the study of fault-tolerant quantum algorithms, where the deterministic advantage should be demonstrated by analyzing an explicit runtime complexity bound and showcasing provable speedups over the best classical counterparts \cite{harrow2017quantum}. Second, our proposal distinguishes itself from the classical shadow \cite{huang2020predicting}, as it cannot calculate fidelity between two mixed states. Namely, the classical shadow is specifically designed to estimate the expectation values of   observables $\{O_i\}$ with $\{\Tr(\rho O_i)\}$. However, the fidelity defined in Eqn.~(\ref{eqn:fide}) can only be reduced to this form when one of the states is a pure state and its classical description is known.

\subsection{Numerical simulations}          
We conduct numerical simulations to evaluate performance of QAE-based fidelity estimator. 

\noindent\textbf{The construction of the employed quantum states.} We synthesize two pairs of states and employ QAE-based Fidelity estimator to calculate their fidelity. The first pair of states contains two eight-qubit state $\rho$ and $\sigma$ with $N=8$. The construction of   $\rho$ and $\kappa$ can be decomposed into two steps. In the first step, we prepare two pure states, i.e., $\ket{\Psi^{(0)}} = \ket{0}^{\otimes 8}$ and $\ket{\Psi^{(1)}} = \bm{\alpha}_1\ket{10...0}+\bm{\alpha}_2\ket{01...0}+...+\bm{\alpha}_8\ket{00...1}$, where the coefficients are fixed to be $\bm{\alpha}_1=\bm{\alpha}_2=...=\bm{\alpha}_8 = 1/8$. Once the pure state $\ket{\Psi}$ (e.g., $\ket{\Psi^{(0)}}$ or $\ket{\Psi^{(1)}}$) is prepared, the noisy channel $\mathcal{N}(\cdot)$ is applied to this state, i.e.,
\begin{equation}\label{eqn:const-fide-states}
	\mathcal{N}(\ket{\Psi}, p, r, a) = \begin{cases}
		\ket{\Psi}, ~\text{with probability}~ p;\\
		\mathbb{V}(r),~\text{with probability}~ 1-p,
	\end{cases}
\end{equation}
where $a\in\mathbb{R}$ is a hyper-parameter, $\mathbb{V}(r)\in\mathbb{R}^{2^8\times 2^8}$ refers to a density operator with $\mathbb{V}_{ii} = 1.5^{-ai}$ when $i\leq r$ and $\mathbb{V}_{ii} = 0$ when $i>r$, and $\mathbb{V}_{ij}=0$ when $i\neq j$.  To construct $\rho$ and $\kappa$, we set  $\rho=\mathcal{N}(\ket{\Psi^{(0)}}, p=0.1, r=8, a= 2)$ and $\kappa=\mathcal{N}(\ket{\Psi^{(1)}}, p=0.5, r=17, a= 5)$, respectively.  In this way,   $\rho$ and $\kappa$ correspond to two eight-qubit states with rank $r^*=8$ and $r^*=17$. The exact fidelity between $\rho$ and $\sigma$ is $\FF(\rho,\sigma)=0.765$.

We follow the same construction rule to produce the second pair of states. The number of qubits is $N=8$ and we set the two states to be the same with the rank $2^8$, i.e., $\rho=\sigma$. The parameter setting of the noisy channel is $p=0.5, r=2^8, a= 1/10$. The explicit form of the two states is $\rho=\sigma=\mathcal{N}(\ket{\Psi^{(0)}}, 0.5, 2^8, 1/10)$. The exact fidelity is $\FF(\rho,\sigma)=1$.

\begin{figure*}
	\centering
\includegraphics[width=0.97\textwidth]{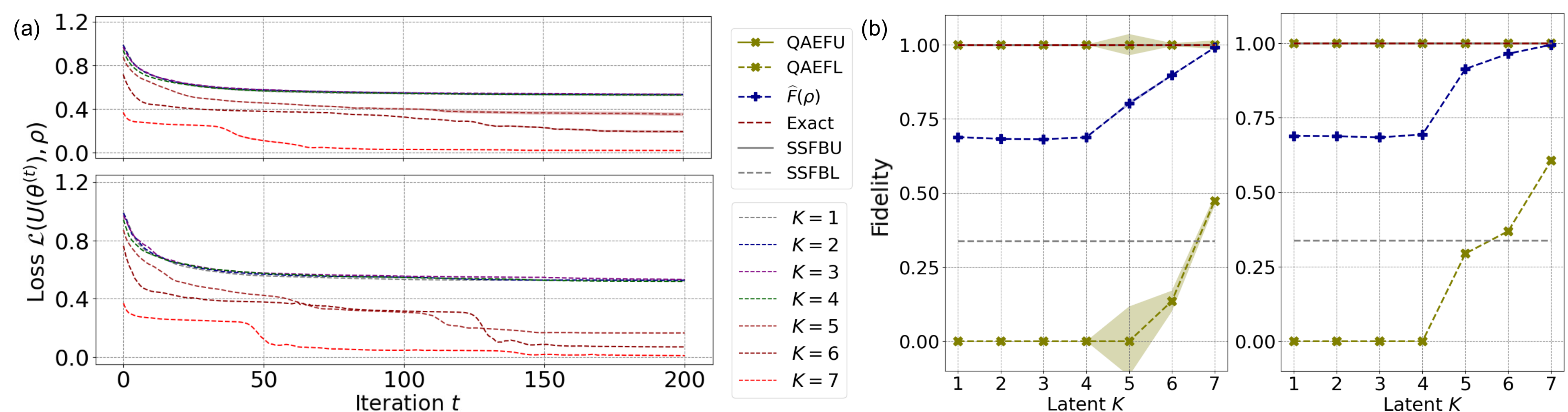}
\caption{\small{\textbf{The simulation results for full-rank states fidelity estimation}. The left panel demonstrates the training loss of the QAE-based fidelity estimator in the first stage when it applies to estimate $\FF(\rho, \rho)$. The right panel illustrates the estimated fidelity bounds and SSFB. The meaning of labels is identical to those in Figure~\ref{fig:QAE} and Figure \ref{fig:QAE-fide-low-append}.}}
\label{fig:QAE-fide-full}
\end{figure*}  
\noindent\textbf{Hyper-parameters settings.}  For both of the two pair states, we vary the size of latent space with $K= 1, ..., 7$. The initialized parameters are sampled from Gaussian distribution with the same mean and the differed variance. Namely, for all settings, the mean is set as zero while the variance is set as $0.4\times (N-K)$. The Adam optimizer is employed to complete all simulations with the learning rate being $0.08$.  The implementation of the quantum encoder $U(\bm{\theta})$ employs the hardware-efficient ansatz  
\begin{equation}\label{eqn:append-MPCQ-def}
	U(\bm{\theta})=\prod_{l=1}^LU_l(\bm{\theta}).
\end{equation}
An intuition is shown in the lower panel of Figure \ref{fig:QAE-fide-low-append}(a). The number of iterations is  $T=200$. Considering that  a higher expressive ansatz is necessary to suppress the input state into a lower dimensional space, the layer number of ansatz depends on the number of latent qubits $K$. That is, for the first pair of states, we set  $L= 11 - K$; for the second pair of states, we set $L=13-K$.  To analyze the robustness of our proposal, each setting described above is repeated with five times. 

\noindent\textbf{Code and device information}. The proposed three QAE-based learning protocols are implemented by Python, with PennyLane as the backbone. We conduct all numerical simulations on  TESLA P40 GPUs using the simulation backends provided by PennyLane. The source code is available at \url{https://github.com/yuxuan-du/Quantum-auto-encoders-based-learning-protocols}. 

\noindent\textbf{The sub-fidelity and super-fidelity bounds.} We employ the sub-fidelity and super-fidelity bounds (SSFB) \cite{miszczak2009sub} as reference to benchmark performance of the QAE-based fidelity estimator. Specifically, the sub-fidelity between   $\rho$ and $\kappa$ yields 
\begin{equation}
	\FF_L(\rho, \kappa) = \Tr(\rho \kappa) + \sqrt{2\left((\Tr(\rho \kappa))^2 - \Tr(\rho \kappa\rho\kappa)\right)}. 
\end{equation}
Moreover, the super-fidelity between   $\rho$ and $\kappa$ yields 
\begin{eqnarray}
 \FF_U(\rho, \kappa) = \Tr(\rho\kappa) + \sqrt{(1 - \Tr(\rho^2))(1-\Tr(\kappa^2))}.  
\end{eqnarray}

\noindent\textbf{Simulation results.} The simulation results for the first pair of states are illustrated in Figure \ref{fig:QAE-fide-low-append}(b)-(e). In particular, the subplots (c) and (e) illustrate the average and top-1 training loss of QAE-based fidelity estimator. The top-1 loss at the last iteration ($T=200$) is dramatically reduced with respect to the increase $K$, i.e., the loss values are $0.415$, $0.035$, $4.26\times 10^{-5}$, $1.6\times 10^{-4}$, $1.82\times 10^{-5}$, $2.92\times 10^{-5}$, and $4.57\times 10^{-11}$ for $K=1,2,...,7$, respectively. From the statistical view, when $K\geq 5$, the average training loss $\mathcal{L}(U(\bm{\theta}), \rho)$ fast converges to zero with $\delta<10^{-3}$. When $K=3,4$, the average training loss is around $0.02$. When $K=1,2$, the dimension of the latent space is lower than the rank of $\rho$ and results into the high training loss, where the average losses are $0.46$ and $0.27$, respectively. The shadow region refers to the variance of training loss, which indicates that QAE-based fidelity estimator is robust, especially for the large $K$. Subplots (d)-(e) separately depict the estimated fidelity under the average and top-1 measures, respectively. The estimated fidelity $\hFide(\rho, \kappa)$ highlighted by the blue solid line fast converges to the exact fidelity $\FF(\rho,\kappa)$ highlighted by the red solid line by increasing $K$. Moreover, for both measures, the estimated bounds $\hFide(\rho, \kappa)\pm \sqrt{\varsigma}$ also fast converge to the true fidelity $\FF(\rho,\kappa)$, as indicated by the green solid and dashed line with cross markers. In the measure of top-1 performance, QAE-based Fidelity estimator achieves better performance than SSFB methods (the solid and dashed grey lines) when $K\geq 3$. The above analysis exhibits the potential of our proposal when the number of latent qubits satisfies $2^{K}\geq r^*$.

The simulation results for the second pair of states are shown in Figure \ref{fig:QAE-fide-full}. The subplot (a) shows the training loss of QAE with respect to the different $K$. An observation is that QAE-based fidelity estimator is robust, since its variance can be ignorable. Moreover, under both the average and top-1 measures, the training loss $\delta$ is lower than $0.02$ only when $K=7$. The degraded performance compared to the first task is caused by the fact that $\rho$ is full rank, which is incompressible by QAE. As such, some information in $\rho$ must be lost in the compressed state $\sigma$ when  $2^K < r^*$. As a result, increasing $K$ enables $\sigma$ to preserve more information in $\rho$ and reduce the training error $\delta$. Subplot (b) depicts the estimated fidelity bounds of the QAE-based fidelity estimator, the SSFB, and the exact result. On average, QAE-based fidelity estimator outperforms SSFB when $K=7$. In the measure of top-1 performance, QAE-based fidelity estimator outperforms SSFB when $K\geq 6$.   
          
 \begin{figure*}
\centering
\includegraphics[width=0.98\textwidth]{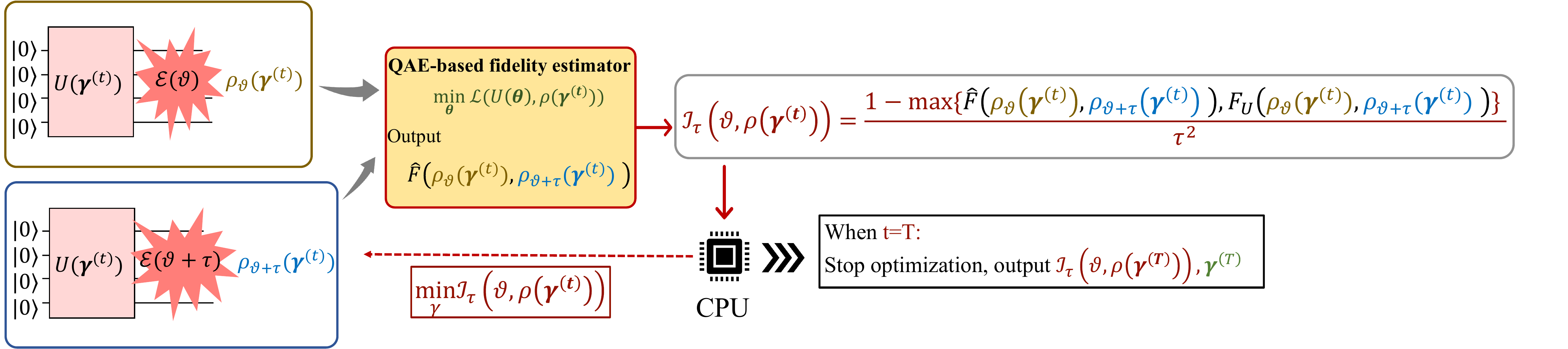}
	\caption{\small{\textbf{The paradigm of the QAE-based QFI estimator}. The QAE-based QFI estimator is implemented on a quantum-classical hybrid system. In the quantum part, a variational quantum circuit $U(\bm{\gamma})$ is employed to prepare the tunable probe state $\rho(\bm{\gamma})$. This probe state is separately interacted with the source described by the parameter $\vartheta$ and $\vartheta+\tau$ to prepare the state $\rho_{\vartheta}(\bm{\gamma})$ and  $\rho_{\vartheta+\tau}(\bm{\gamma})$. The two generated states are fed into the QAE-based fidelity estimator to estimate $\FF(\rho_{\vartheta}(\bm{\gamma}), \rho_{\vartheta+\tau}(\bm{\gamma}))$, as highlighted by the red box. The calculated fidelity is used to compute 
    	$ \mathcal{I}_{\tau}(\vartheta, \rho_{\vartheta}(\bm{\gamma}))$ in Eqn.~(\ref{eqn:q_metro_def}) and then update parameters $\bm{\gamma}$ to maximize this quantity. Repeating the above procedures with $T$ times, the classical processor outputs $\gamma^{(T)}$, which can be used to prepare the estimated probe state $\rho(\bm{\gamma}^{(T)})$.  }}
  \label{fig:QAE-QFI}
\end{figure*}
             
\section{QAE-based QFI estimator}\label{sec:QFI-est}
The high precision of sensing is of great importance to many scientific and real-world applications, which include but not limit to chemical structure identification \cite{lovchinsky2016nuclear} and  gravitational wave detection \cite{mcculler2020frequency}. However, the central limit theorem indicates that the statistical error in the classical scenario must be lower bounded by $\nu^{-1/2}$, where  $\nu$ is the number of measurements. Driven by the significance of improving the precision of sensing, it highly desired to utilize quantum effects to further reduce the statistical error.

Quantum metrology provides a positive response towards the above questions \cite{giovannetti2006quantum,giovannetti2011advances,koczor2020variational}. Through leveraging coherence or entangled quantum states, the statistical error can be reduced by an amount proportional to $\nu^{-1}$.   A fundamental quantity in quantum metrology is \textit{quantum Fisher information} (QFI) \cite{petz2011introduction}, which determines the best possible estimation error for a certain quantum state. Mathematically, let $\rho_{\vartheta}(\bm{\gamma})$ be the evolved state, which is prepared by interacting a tunable $N$-qubit state $\rho(\bm{\gamma})$ with a source that encodes the environmental information in a single parameter $\vartheta$. The QFI $\mathcal{I}(\vartheta;\rho_{\vartheta}(\bm{\gamma}))$ measures the ultimate precision $\Delta \vartheta$, i.e., $(\Delta \vartheta)^2\geq 1/(\nu \mathcal{I}(\vartheta; \rho_{\vartheta}(\bm{\gamma})))$, where $\nu$ is the number of measurements to estimate $\vartheta$. At the expense of computational simplicity, $\mathcal{I}(\vartheta;\rho_{\vartheta}(\bm{\gamma}))$ is   estimated by 
\begin{equation}\label{eqn:q_metro_def}
	\mathcal{I}_{\tau}(\vartheta;\rho_{\vartheta})= 8\frac{1 - \FF(\rho_{\vartheta}, \rho_{\vartheta+\tau})}{\tau^2}. 
\end{equation}  
When $\tau\rightarrow 0$, the approximation error suppresses to $0$, i.e., $\mathcal{I}_{\delta}(\vartheta;\rho_{\vartheta}(\bm{\gamma})) = \mathcal{I}_{\tau\rightarrow 0}(\vartheta;\rho_{\vartheta}(\bm{\gamma}))$. Intuitively, a true state $\rho_{\vartheta}(\bm{\gamma})$ with a high QFI $\mathcal{I}_{\tau}(\vartheta;\rho_{\vartheta}(\bm{\gamma}))$ is sharply different from the error state $\rho_{\vartheta+\tau}(\bm{\gamma})$, which implies that the parameter $\vartheta$ can be effectively estimated via measurement.

As indicated in Lemma \ref{lem:DQC-1}, the calculation of the fidelity $\FF(\rho_{\vartheta}(\bm{\gamma}), \rho_{\vartheta+\tau}(\bm{\gamma}))$ in Eqn.~(\ref{eqn:q_metro_def}) is classically computational hard for general density matrices. Such an inefficiency prohibits the applicability of quantum metrology. To address this issue, Ref.~\cite{beckey2020variational} generalized the fidelity estimation solver \cite{cerezo2020variational} to estimate $\mathcal{I}_{\tau}(\vartheta;\rho_{\vartheta}(\bm{\gamma}))$ with potential runtime advantages. Their proposal employed variational quantum circuits to prepare  states $\rho_{\vartheta}(\bm{\gamma})$ and $ \rho_{\vartheta+\tau}(\bm{\gamma})$ by adjusting $\bm{\gamma}$. Once these two quantum states are prepared, the variational fidelity estimator \cite{cerezo2020variational} is applied to estimate $\mathcal{I}_{\tau}(\vartheta;\rho_{\vartheta}(\bm{\gamma}))$. The optimization is conducted to find optimal $\bm{\gamma}^*$ that maximizes $\mathcal{I}_{\tau}(\vartheta;\rho_{\vartheta}(\bm{\gamma}))$. 

Here we combine variational quantum circuits with the QAE-based fidelity estimator to construct a more advanced QFI estimator. Recall that to balance the computational efficiency and the learning performance, the variational QFI estimator in \cite{beckey2020variational} exploits the variational fidelity estimator to estimate $\FF(\rho_{\vartheta}(\bm{\gamma}), \rho_{\vartheta+\tau}(\bm{\gamma}))$. In this regard, we can employ QAE-based fidelity estimator to devise the QAE-based QFI estimator, which can be integrated into advanced quantum metrology protocols to obtain precise estimation using few measurements.

\subsection{Implementation of the QAE-based QFI estimator}
The implementation of the QAE-based QFI estimator is illustrated in Figure \ref{fig:QAE-QFI}. In particular, our proposal repeats the following procedure with in total $T$ times. At the $t$-th iteration, the variational quantum circuit $U(\bm{\gamma}^{(t)})$ is utilized to generate the probe state $\rho(\bm{\gamma}^{(t)})$. As highlighted by the brown and blue boxes, the probe state interacts with the target source described by the parameter $\vartheta$ and $\vartheta+\tau$ to generate the states  $\rho_{\vartheta}(\bm{\gamma})$ and  $\rho_{\vartheta+\tau}(\bm{\gamma})$, respectively. These two states, which contain the information of the source $\vartheta$, are fed into the QAE-based fidelity estimator presented in Section \ref{sec:QAE-fide}. The estimated fidelity returned by the QAE-based fidelity estimator is employed to compute $ \mathcal{I}_{\tau}(\vartheta, \rho_{\vartheta}(\bm{\gamma}))$ in Eqn.~(\ref{eqn:q_metro_def}) and then update parameters $\bm{\gamma}$ to maximize this quantity. This completes the $t$-th iteration. 

The following lemma quantifies how the training error of QAE influences the performance of the proposed estimator, where the corresponding proof is given in Appendix \ref{append:sec:lem-QFI}. 
\begin{lem}\label{lem:QFI_QAE}
	Following notations in Eqns.~(\ref{eqn:loss_G1}) and  (\ref{eqn:q_metro_def}), let $\tau$ be the shift parameter. When QAE is applied to estimate the fidelity $\FF(\rho_{\vartheta}(\bm{\gamma}), \rho_{\vartheta+\tau}(\bm{\gamma}))$ with the estimation quantity $\varsigma$ in Eqn.~(\ref{eqn:est-fide-error-quantity}). Then the QAE-based QFI estimator returns   $\hat{\mathcal{I}}_{\tau}(\vartheta;\rho_{\vartheta})$ with 
	\begin{equation}
 	 	  \left| \hat{\mathcal{I}}_{\tau}(\vartheta, \rho_{\vartheta}(\bm{\gamma})) - \mathcal{I}_{\tau}(\vartheta, \rho_{\vartheta}(\bm{\gamma})) \right|  \leq   \frac{8\sqrt{2\varsigma}}{\tau^2}. 
 	\end{equation}	
\end{lem}
\noindent The above results imply that the estimation error of QFI is controlled by the quantity $\varsigma$. Moreover, the relation between the training loss $\delta$ and $\varsigma$ elucidated in Section~\ref{sec:implementation-QAE-FE} hints that the performance of the QAE-based QFI estimator can be improved by adopting problem-dependent ansatzes and advanced optimizers by suppressing $\delta$.  Last, supported by Lemma \ref{lem:DQC-1},  the necessary condition of our proposal to reach the potential merits is  $\sqrt{2\varsigma} <\text{poly}(1/N)$ when the total runtime is restricted to $O(\text{poly}(N))$.

\subsection{Numerical simulations}

\noindent\textbf{Problem setup}.  We apply the QAE-based QFI estimator to prepare the probe state $\rho(\bm{\gamma})$ that maximizes the QFI for a magnetic field described by the parameter $\vartheta$.  Mathematically, the probe state interacting with the magnetic field transforms to $\rho_{\vartheta}(\bm{\gamma})= W(\vartheta)\rho(\bm{\gamma})W(\vartheta)^{\dagger}$, where $W(\vartheta) =  e^{-i\vartheta G}$ is a unitary transformation, $G=\sum_{i=1}^J Z_i$ is the problem Hamiltonian, and $Z_i$ is the Pauli-Z matrix on the $i$-th qubit. We set $J=4$ and $\vartheta = 0.1$ to proceed the following simulations.

 \begin{figure}[h]
	\centering
\includegraphics[width=0.38\textwidth]{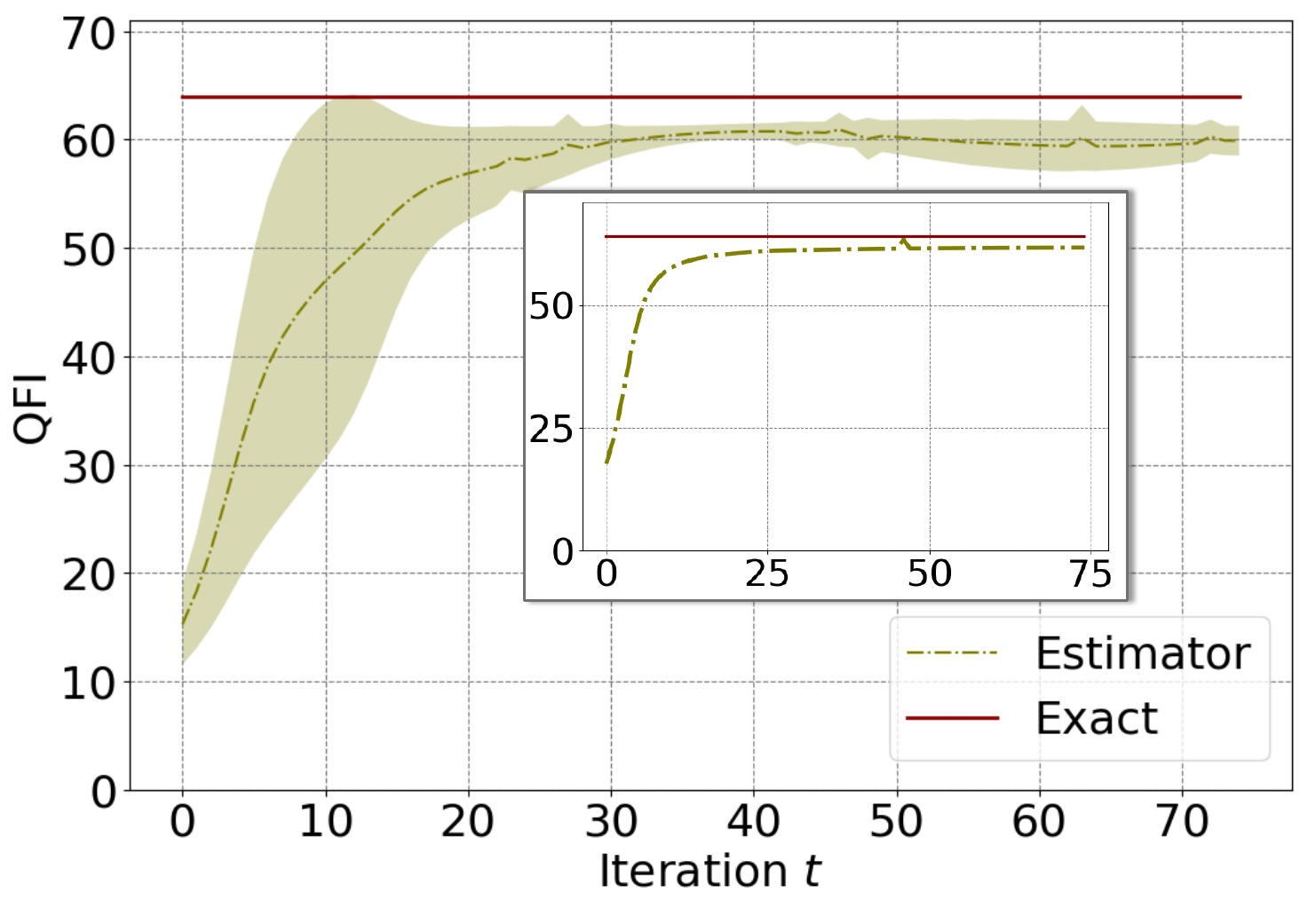}
\caption{\small{\textbf{The simulation results of the QAE-based QFI estimator}. The outer and inner plots depict the average and top-1 performance of the proposed QAE-based QFI estimator. The shaded region represents the standard deviation. The labels `Exact' and `Estimator' separately represent the maximum QFI that can be achieved and the estimated QFI returned by the proposed estimator.  }}
\label{fig:sim-QFI}
\end{figure}

\begin{figure*}
	\centering
\includegraphics[width=0.78\textwidth]{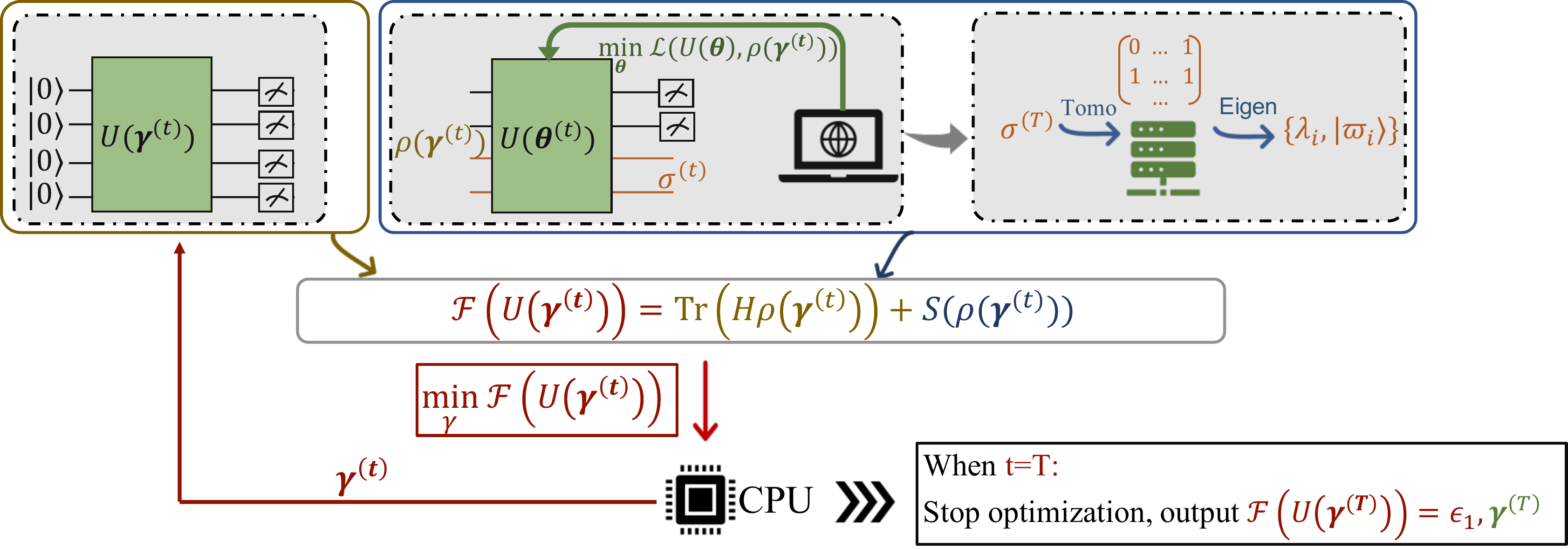}
	\caption{\small{\textbf{The paradigm of the QAE-based quantum Gibbs state solver}. The setup is composed of three components. The first component is the variational quantum circuit $U(\bm{\gamma}^{(t)})$ to prepare the variational quantum Gibbs state $\rho(\bm{\gamma}^{(t)})$. The second component is QAE, which is used to extract the spectral information of $\rho(\bm{\gamma}^{(t)})$. The third component is the classical processor to optimize the variational quantum circuit $U(\bm{\gamma}^{(t)})$ and the quantum encoder $U(\bm{\theta}^{(t)})$. See Section~\ref{sec:QAE-gibbs} for the learning procedure.  
	}}
	\label{fig:QAE-gibbs-scheme}
\end{figure*}

\noindent\textbf{Hyper-parameters settings.} We adopt the following hyper-parameter settings to conduct all simulations.  Specifically, the variational quantum circuits $U(\bm{\gamma})$ and $U(\bm{\theta})$, which are used to prepare the probe state $\rho_{\vartheta}(\bm{\gamma})$ and to build the quantum encoder respectively, follow the hardware-efficient ansatz in Eqn.~(\ref{eqn:append-MPCQ-def}). The layer number $L$ for $U(\bm{\gamma})$ and $U(\bm{\theta})$ is set as $6$ and $5$, respectively. A slight modification in building $U(\bm{\gamma})$ is removing all $R_Z$ gates in the single-qubit layer to accelerate training. For QAE, the number of latent qubits is set as $2$, the number of iterations is $300$, and in each iteration, the Adam optimizer is used to update $\bm{\theta}$ with the learning rate being $0.02$. For the probe state, the total number of iterations is $75$ and in each iteration, the parameters $\bm{\gamma}$ are updated by parameter shift rule with the learning rate being $0.08$. Each setting is repeated with five times to obtain the statistical results.

\noindent\textbf{Simulation results.} The simulation results are exhibited in Figure \ref{fig:sim-QFI}. Recall that a well-known conclusion in quantum metrology is that the optimal probe state is an $N$-qubit GHZ state in the noiseless scenario, where the QFI reaches the Heisenberg limit $4N^2$. Compared with the optimal result $64$, the average QFI of the optimized probe state oscillates around $60$ after $30$ iterations with the standard deviation up to $3$. The inner plot illustrates the top-1 performance of our proposal. In the $46$-th step, its QFI reaches  $63.317$.  The aforementioned observation highlights the feasibility of utilizing QAE-based QFI estimator for preparing the probe state, and suggests an intriguing avenue for future research to enhance the robustness and stability of the proposed QAE-based QFI estimator.

\section{QAE-based quantum Gibbs state solver} \label{sec:QAE-gibbs}

Gibbs state describes the thermal equilibrium properties of quantum systems. Efficient preparation of quantum Gibbs states is a central problem in quantum simulation, quantum optimization, and quantum machine learning \cite{poulin2009sampling}. As such, an efficient quantum Gibbs state preparation solver can benefit many  applications. For example, the quantum semi-definite programs exploit quantum Gibbs states to earn runtime speedups \cite{brandao2017quantum}.  The mathematical  expression of quantum Gibbs state  is 
\begin{equation}\label{eqn:def-gibbs}
\rho_G = \frac{e^{-\beta H}}{\Tr(e^{-\beta H})},	
\end{equation}
where $\beta$ refers to the inverse temperature and $H\in\mathbb{C}^{2^N\times 2^N}$ is a specified $N$-qubit Hamiltonian. However, theoretical results based on computational complexity theory \cite{Watrous2009} indicate that the preparation of quantum Gibbs states at low-temperature is difficult, i.e., computing the partition function $\Tr(e^{-\beta H})$ in two-dimensions is $\NP$-hard \cite{schuch2007computational}. Driven by the significance and the intrinsic hardness of preparation, Ref.~\cite{wang2020variational} explored how to use variational quantum circuits implemented on NISQ devices to generate a state $\rho(\bm{\gamma}^*)\in\mathbb{C}^{2^N\times 2^N}$ approximating $\rho_G$. In particular,   $\rho(\bm{\gamma}^*)$ minimizes the free energy $\mathcal{F}$ such that 
\begin{equation}\label{eqn:free-energy}
 \rho(\bm{\gamma}^*)=\arg\min_{\rho(\bm{\gamma})}\mathcal{F}(\rho(\bm{\gamma})),	
\end{equation} 
where $\mathcal{F}(\rho(\bm{\gamma}))=\Tr(H\rho(\bm{\gamma}))-\beta^{-1}S(\rho(\bm{\gamma}))$ and $S(\rho(\bm{\gamma}))=-\sum_{i}\lambda_i\log \lambda_i$ refers to the Von-Neumann entropy \cite{nielsen2010quantum}. To further improve the computational efficiency, Ref.~\cite{wang2020variational} replaces $S(\rho(\bm{\gamma}))$ with its truncated version, i.e., $S_R(\rho(\bm{\gamma}))= \sum_{j=0}^R C_j\Tr(\rho(\bm{\gamma})^{j+1})$ with $C_j$ being certain constants. When $S(\rho(\bm{\gamma}))$ is estimated by $S_2(\rho(\bm{\gamma}))$, the quantity used to optimize   $\rho(\bm{\gamma})$ yields $\arg\min_{\rho(\bm{\gamma})} \mathcal{F}(\rho(\bm{\gamma})):= \Tr(H\rho(\bm{\gamma}))-\frac{2\Tr(\rho(\bm{\gamma})^2)}{\beta} -\frac{\Tr(\rho(\bm{\gamma})^3) + 3}{2\beta}$.  Enlightened by the fact that $S(\rho(\bm{\gamma}))$ amounts to the calculation for the eigenvalues of $\rho(\bm{\gamma})$,  we apply QAE to estimate $S(\rho(\bm{\gamma}))$. Mathematically, the approximation of the free energy takes the form as
\begin{equation}\label{eqn:QAE-free-energy}
\hat{\mathcal{F}}(\rho(\bm{\gamma}))= \Tr(H\rho(\bm{\gamma})) -\beta^{-1}S(\hat{\rho}(\bm{\gamma})),
\end{equation}
where $S(\hat{\rho}(\bm{\gamma}))$ refers to the Von-Neumann entropy of the reconstructed state $\hat{\rho}(\bm{\gamma}$ returned by QAE with $S(\hat{\rho}(\bm{\gamma}))=-\sum_{i}\hat{\lambda}_i\log \hat{\lambda}_i$. This strategy allows a dramatically reduction of  the required quantum resources but the estimation accuracy can be improved. In the optimal case, the spectral information of the compressed state $\sigma^*$, i.e., a set of eigenvalues $\{\lambda_i\}_{i=1}^{r^*}$, is sufficient to recover  $S(\rho(\bm{\gamma}))$. 

\subsection{Implementation  of the QAE-based Gibbs state solver}
We now elaborate on the implementation  of the proposed QAE-based quantum Gibbs state solver shown in Figure \ref{fig:QAE-gibbs-scheme}. Suppose that the total number of iterations to optimize the free energy $\hat{\mathcal{F}}(\rho(\bm{\gamma}))$ is $T$. The mechanism of our proposal at the $t$-th iteration with $t\in[T]$ is as follows. 

\begin{enumerate}
	\item The variational quantum Gibbs state $\rho(\bm{\gamma}^{(t)})$ prepared by the variational quantum circuit $U(\bm{\gamma}^{(t)})$, i.e., $\rho(\bm{\gamma}^{(t)}) = \Tr_{A}(U(\bm{\gamma}^{(t)})(\ket{\bm{0}}\bra{\bm{0}})U(\bm{\gamma}^{(t)})^{\dagger})$, is interacted  with $H$ to compute $\Tr(H\rho(\bm{\gamma}^{(t)}))$, as highlighted in the brown box.
	\item  QAE is employed to compress the state $\rho(\bm{\gamma}^{(t)})$ prepared by $U(\bm{\gamma}^{(t)})$. Once the optimization of QAE is completed, the compressed state $\sigma^{(T)}$ is extracted into the classical form followed by the spectral decomposition $\{\hat{\lambda}_i\}$.  This step is highlighted in the blue box. 
	\item  Given access to $S(\hat{\rho}(\bm{\gamma}^{(t)}))$  and $\Tr(H\rho(\bm{\gamma}^{(t)}))$, the classical optimizer can minimize $\hat{\mathcal{F}}(\rho(\bm{\gamma}^{(t)}))$. The minimization  is achieved by the gradient-based methods,  i.e., the update rule at the $t$-th iteration  is 
\begin{equation}\label{eqn:QAE-update-gibbs}
	\bm{\gamma}^{(t+1)} = 	\bm{\gamma}^{(t)} - \eta \frac{\partial \hat{\mathcal{F}}(\rho(\bm{\gamma}^{(t)}))}{\partial \bm{\gamma}^{(t)}}.
\end{equation} 
\end{enumerate}  
After repeating the above procedures with in total $T$ iterations, the QAE-based quantum Gibbs state solver outputs the optimized free energy $\hat{\mathcal{F}}(\rho(\bm{\gamma}^{(T)}))$ and the trained parameters $\bm{\gamma}^{(T)}$.

We quantify the lower bound of the fidelity between the target state $\rho_G$ and the state $\rho(\hat{\bm{\gamma}})$ prepared by the QAE-based   quantum Gibbs state solver in the following lemma. The corresponding proof is provided in Appendix \ref{append:subsec:gibbs-lem1}.   
\begin{lem}\label{thm:QAE-Gibbs}
	Following notation in Eqn.~(\ref{eqn:free-energy}), let $\hat{\bm{\gamma}}$ be the optimized parameters of the QAE-based quantum Gibbs state solver. Suppose $|\hat{\mathcal{F}}(\rho(\hat{\bm{\gamma}}))-\mathcal{F}(\rho_G)|\leq \epsilon_1$ as the discrepancy between the optimized and the optimal free energy is less than the threshold $\epsilon_1$. Let $\delta$ be the training loss of QAE when compressing the state $\rho(\hat{\bm{\gamma}})$ and $\varsigma$ be the quantity defined in Eqn.~(\ref{eqn:est-fide-error-quantity}). The fidelity between the approximated state $\rho(\hat{\bm{\gamma}})$ returned by QAE-based quantum Gibbs state solver and the target state $\rho_G$ is lower bounded by
	\begin{equation}
	 	\FF(\rho_G, \rho(\hat{\bm{\gamma}})) \geq  1 - \sqrt{2\varsigma \left(1 + \log \frac{2^N-1}{1-\sqrt{1-\delta}} \right) + 2\beta \epsilon_1 }. 
	\end{equation} 
\end{lem}    
\noindent The achieved results demonstrate that even though the optimization of QAE is not optimal, the collected spectral information of the compressed state can also be used to quantify the Von-Von-NeumannNeumann entropy with the estimation guarantee, while the price to pay is decreasing the fidelity between $\rho(\hat{\bm{\gamma}})$ and $\rho_G$.  Furthermore, when the inverse temperature $\beta$, the gap $\epsilon_1$, and the training loss of QAE $\delta$ are small, the prepared state  $\rho(\hat{\bm{\gamma}})$ can well approximate the target state $\rho_G$ in Eqn.~(\ref{eqn:def-gibbs}). In this respect, as with prior two protocols, the key to gain computation advantages is exploiting advanced optimization techniques to improve the trainability of QAE and hence the performance of the QAE-based quantum Gibbs state solver.

\subsection{Numerical simulations}

We apply our proposal to prepare the Gibbs state corresponding to the Ising chain model and evaluate its performance.

\noindent\textbf{Construction rule of the target state.} The problem Hamiltonian in Eqn.~(\ref{eqn:def-gibbs}) is   $H_B= -\sum_{i=1}^J Z_iZ_{i+1}$ with $Z_{J+1}=Z_{1}$ and $Z_i$ is the Pauli-Z matrix on the $i$-th qubit. The quantum Gibbs state takes the form as $\rho_G = \frac{e^{-\beta H_B}}{\Tr(e^{-\beta H_B})} \in \mathbb{C}^{2^3\times 2^3}$. We adopt this construction rule to generate three target Gibbs states with $J=3$ and $\beta\in \{1.2,1.5,4\}$.

\noindent\textbf{Hyper-parameters setting.} 
We apply the same hyper-parameters setting to learn three Gibbs states. In particular, the variational Gibbs state $\rho(\bm{\gamma})$ is prepared by interacting $U(\bm{\gamma}^{(t)})$ with a $4$-qubit state $\ket{0}_A\ket{000}_B$ (the subscripts $A$ and $B$ refer to two subsystems), followed by tracing out the subsystem $A$, i.e., 
	$\rho(\bm{\gamma}^{(t)}) = \Tr_A(U(\bm{\gamma}^{(t)})(\ket{0}_A\ket{000}_B\bra{0}_A\bra{000}_B)U(\bm{\gamma}^{(t)})^{\dagger}).$ The trainable unitary $U(\bm{\gamma}^{(t)})$ is implemented by the hardware-efficient ansatz in Eqn.~(\ref{eqn:append-MPCQ-def}) with $L=5$.  The total number of iterations to optimize $\mathcal{F}(\rho(\bm{\gamma}^{(t)}))$ in Eqn.~(\ref{eqn:free-energy}) is set as $T=200$. The parameter shift rule is applied to update $\bm{\gamma}$ with the learning rate $\eta=0.2$. In each iteration $t$, QAE is used to acquire the spectral information of $\rho(\bm{\gamma}^{(t)})$. The hardware-efficient ansatz in Eqn.~(\ref{eqn:append-MPCQ-def}) is exploited to implement the quantum encoder  $U(\bm{\theta})$ with $L=4$. The number of latent qubits is set as $K=2$. The Adam optimizer is used to update $\bm{\theta}$ with the learning rate $\eta=0.02$. The   number of iterations to optimize $\bm{\theta}$ is set as $100$. The initialized parameters for both $\bm{\gamma}$ and $\bm{\theta}$ are uniformly sampled from the interval $[0, 1]$. We repeat each setting with five times to collect the statistical results.

\begin{figure}
	\centering
\includegraphics[width=0.385\textwidth]{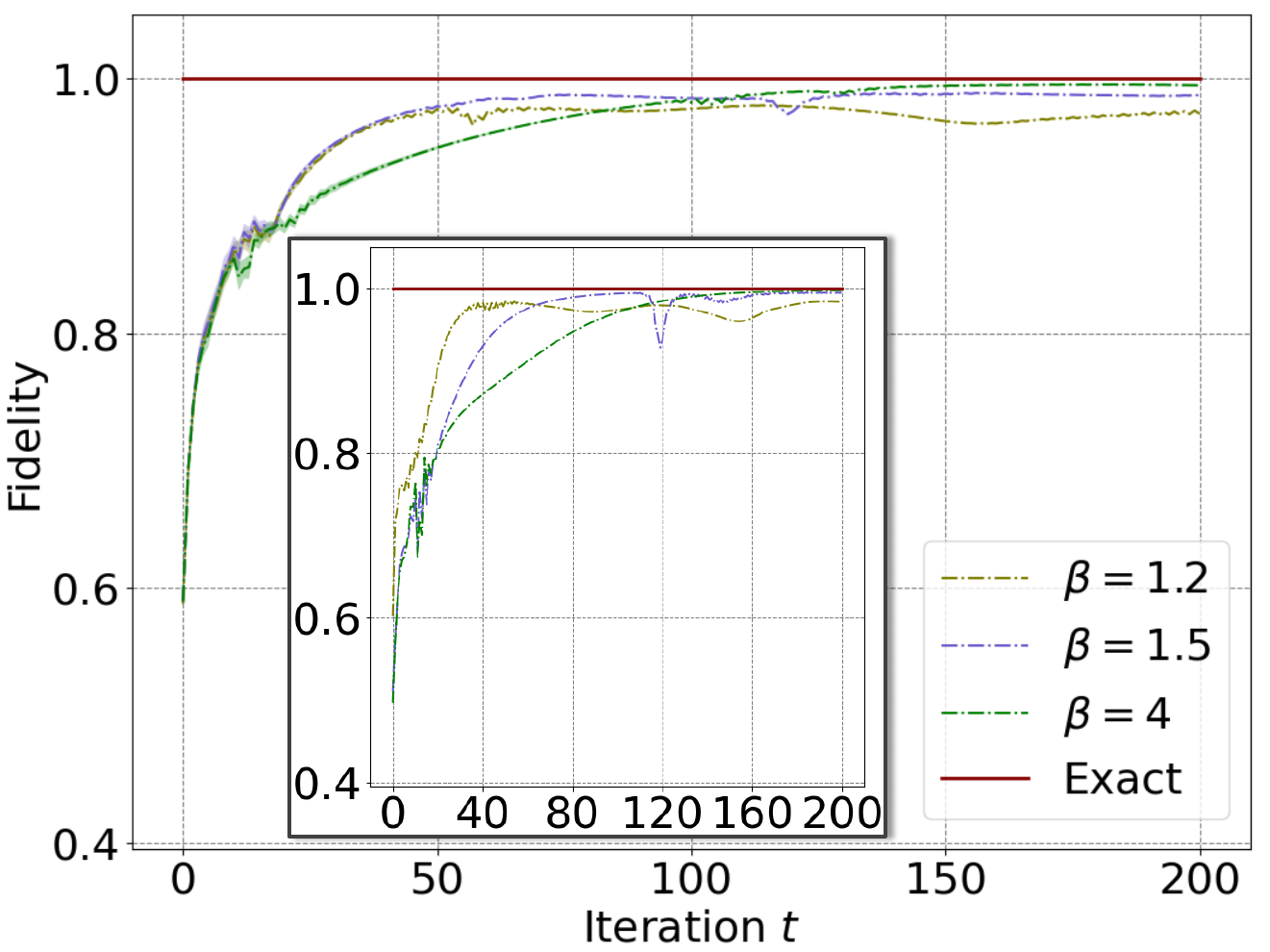}
\caption{\small{\textbf{The fidelity $\FF(\rho(\bm{\gamma}),\rho_G)$ with the varied $\beta$}. The label `$\beta=x$' represents how the average fidelity $\FF(\rho(\bm{\gamma}),\rho_G)$ evolves in $T=200$ iterations when the inverse temperature $\beta$ is set as $x$. The shadow region represents variance. The inner plot shows the top-1 performance of each setting.  }}
\label{fig:gibbs}
\end{figure}

\noindent\textbf{Simulation results.} The simulation results are illustrated in Figure \ref{fig:gibbs}. The outer plot exhibits the the average  fidelity between the prepared quantum Gibbs state $\rho(\bm{\gamma})$ and the target state $\rho_G$, i.e., for all settings  $\beta\in\{1.2,1.5,4\}$, the average fidelity is above $0.97$ after $200$ iterations. More concisely, the average fidelity $\FF(\rho(\bm{\gamma}),\rho_G)$ converges to $0.972$, $0.987$, and $0.994$ for the setting $\beta=1.2$, $1.5$, and $4$, respectively. The ignorable shaded region indicates the robustness of QAE-based Gibbs state solver. The inner plot depicts the top-1 performance of QAE-based Gibbs state solver with respect to $\beta=1.2$, $1.5$, and $4$. The corresponding fidelity is $0.983$, $0.994$, and $0.997$, respectively. The achieved high fidelity for all settings reflects the effectiveness of the proposed QAE-based  quantum Gibbs state solver. Moreover, the degraded performance with respect to the decreased $\beta$ accords  with Lemma \ref{thm:QAE-Gibbs}, i.e., a smaller $\beta$ leads to a worse fidelity bound. 

\section{Implementation on real quantum hardware}\label{sec:prac_implement} 
We now discuss the feasibility of implementing the proposed QAE-based learning protocols on NISQ machines.  QAE-based fidelity estimator requires access to quantum resources in Step (i) and Step (iii). In Step (i), QAE is employed to compress the input state, which can be accomplished by NISQ machines due to the flexibility of QAEs. Experimental studies have demonstrated the feasibility of executing QAEs on NISQ machines \cite{pepper2019experimental,lamata2018quantum}. In Step (iii) (or the evaluation of $\varsigma$ in Eqn.~(\ref{eqn:est-fide-error-quantity})), a quantum computer executes the SWAP test on the explored state $\kappa$  (or $\rho$)  and the reconstructed state $\ket{\bar{\omega}_i}\ket{0}^{\otimes N-K}$ \cite{subacsi2019entanglement}. Various platforms such as photons and superconducting quantum chips have experimentally demonstrated the implementation of SWAP test circuit \cite{cai2015entanglement,travnivcek2019experimental}. To conduct SWAP test, the eigenvector $\ket{\bar{\omega}_i}$ should be encoded into a $K$-qubit quantum state, which can be efficiently achieved using quantum state preparation methods since the number of latent qubits $K$ is small \cite{nielsen2010quantum,plesch2011quantum}.

We next discuss how to implement QAE-based QFI estimator and QAE-based Gibbs state solver on NISQ machines. The quantum computer plays a similar role in constructing these two learning protocols. First, quantum computer is employed to carry out on a variational quantum circuit, i.e., $U(\bm{\gamma})$, to prepare the  probe state and the variational Gibbs state, respectively. As with other variational quantum algorithms, we can employ hardware-efficient ansatz to implement $U(\bm{\gamma})$, since many experimental studies have  demonstrated the effectiveness of  hardware-efficient ansatz on NISQ machines \cite{kandala2017hardware}. Second, the prepared variational quantum states is fed into QAE to obtain the corresponding spectral information. As explained in the construction of QAE-based fidelity estimator, this procedure can also be efficiently  accomplished by NISQ machines. Taken together, QAE-based QFI estimator and Gibbs state solver can be realized on NISQ machines.
      
\section{Discussion and conclusion}\label{sec:conclusion}
In this study, we revisit QAE and demonstrate how to exploit its spectral properties to benefit quantum computation, quantum information, and quantum sensing tasks. We provide rigorous error bounds for the QAE-based learning protocols, which are crucial for utilizing QAE-based methods to solve practical tasks in quantum system learning. Additionally, we conduct extensive simulations to validate the effectiveness of our proposals. Theoretical analysis implies that the performance of the proposed QAE-based learning protocols can be further improved by adopting advanced optimization methods and ansatz to suppress the training loss. Empirical results validate our theoretical results.

There are two promising future research directions to further advance our protocols. The first future research direction is exploring advanced learning strategies that can effectively suppress the training loss  \cite{stokes2020quantum,wierichs2020avoiding,huang2022learning,wang2022symmetric}. The second future research direction is exploring effective initialization methods for the QAE-based quantum Fisher estimator and QAE-based quantum Gibbs solver \cite{cervera2021meta,jain2022graph,rudolph2023synergistic}. Without a suitable initialization, the prepared variational quantum state may become high rank, leading to a significant estimation error in the training process and resulting in inferior performance.
 
In the end, we would like to point out the potential of applying advanced dimension reduction techniques \cite{lin2010multiple,vidal2005generalized,ye2005idr,ye2006feature} to facilitate many other complicated learning problems in the quantum world such as Hamiltonian identification \cite{Yu2023Hamiltonian} and quantum state reconstruction \cite{Ghosh2021Reconstruct}. Besides, in quantum machine learning \cite{du2021efficient,huang2020experimental,qian2021dilemma,schuld2019quantum,shi2021quantum,wu2021expressivity,yu2019improved}, theoretical studies \cite{huang2021power,wang2021towards} proved that quantum kernels \cite{havlivcek2019supervised,hubregtsen2021training,kusumoto2021experimental,schuld2021supervised} may achieve better generalization ability than that of classical kernels. Considering that the core of quantum kernels is calculating the similarity of two quantum inputs, which amounts to accomplishing the fidelity estimation, efficient dimension reduction techniques  can boost the performance of quantum kernels.


\newpage

\appendices

\section{Proof of Theorem \ref{thm_QAE_noiseless}}\label{Appendix:thm1_opt_res}
\begin{proof}[Proof of Theorem \ref{thm_QAE_noiseless}]
Note that the loss  $\mathcal{L}$ defined in  Eqn.~(\ref{eqn:loss_G1}) is lower bounded by $0$, i.e.,
\begin{eqnarray}
&&\mathcal{L}(U(\bm{\theta}),\rho) \nonumber \\
= && \Tr\left(M_GU(\bm{\theta})\rho U(\bm{\theta})^{\dagger} \right) \nonumber \\
 = && \Tr\left(\left(\mathbb{I}_N- (\ket{0}\bra{0})^{\otimes(N-K)}\otimes \mathbb{I}_K  \right)U(\bm{\theta})\rho U(\bm{\theta})^{\dagger} \right) \nonumber\\
 \geq && 0,
\end{eqnarray}
 since both the measurement operator $M_G$ and the state $U(\bm{\theta})\rho U(\bm{\theta})^{\dagger}$ are positive semi-definite matrices. Such a result implies that the optimal parameters $\bm{\theta}^*$, or equivalently $U(\bm{\theta}^*)$, satisfy
 \begin{equation}
 	\mathcal{L}(U(\bm{\theta}^*),\rho)=0.
 \end{equation}

We next derive the explicit form of $U(\bm{\theta}^*)$ that enables $\mathcal{L}(U(\bm{\theta}^*),\rho)=0$. Recall that the loss function $\mathcal{L}(U(\bm{\theta}),\rho)$ can be rewritten as
\begin{equation}	\label{eqn:proof-thm1-1}
  \mathcal{L}(U(\bm{\theta}),\rho)  
	=   1 - \Tr\left(\left(\ket{0}\bra{0})^{\otimes(N-K)}\otimes \mathbb{I}_K \right)U(\bm{\theta})\rho U(\bm{\theta})^{\dagger} \right).
\end{equation}
To obtain $ \mathcal{L}(U(\bm{\theta}^*),\rho)=0$, it necessitates to ensure
\begin{equation}\label{eqn:thm1-loss-rew}
	\Tr\left(\left(\ket{0}\bra{0})^{\otimes(N-K)}\otimes \mathbb{I}_K  \right)U(\bm{\theta}^*)\rho U(\bm{\theta}^*)^{\dagger} \right) = 1.
\end{equation} 

Considering that any unitary operator is full rank and all its singular values equal to $1$, without loss of generality, we denote the singular value decomposition of the optimal quantum encoder  as 
\begin{equation}\label{eqn:opt_U_thm1} 
	U(\bm{\theta}^*) = \sum_{i=1}^{r^*}  \ket{0^{\otimes (N-K)},\varpi_i}\bra{\varphi_i} + \sum_{i=r^*+1}^{2^N}   \ket{\phi_i}\bra{\varphi_i}, 
\end{equation} 
where $r^*=2^K$, and $\{\ket{0}^{\otimes (N-K)}\ket{\varpi_i}\}\bigcup \{\ket{\phi_i}\}$  (or $\{\ket{\varphi_i}\}$) is a set of orthonormal basis to form the left (or right) singular vectors of $U(\bm{\theta}^*)$.

In conjunction with the above two equations, we obtain
\begin{eqnarray}\label{eqn:thm1-proof-2}
	&& \Tr\left(\left(\ket{0}\bra{0})^{\otimes(N-K)}\otimes \mathbb{I}_K  \right) U(\bm{\theta}^*)\rho U(\bm{\theta}^*)^{\dagger} \right) \nonumber\\
	= && \Tr\Big(\sum_{i=1}^{r^*}  \ket{0^{\otimes (N-K)},\varpi_i}\bra{\varphi_i}    	\rho U(\bm{\theta}^*)^{\dagger} \Big) \nonumber\\
	= && \Tr\Big(\sum_{j=1}^{r^*}  \ket{\varphi_j}\bra{0^{\otimes (N-K)},\varpi_i} \sum_{i=1}^{r^*}  \ket{0^{\otimes (N-K)},\varpi_i}\bra{\varphi_i}    	\rho   \Big) \nonumber\\
	= && \Tr\Big( \sum_{i=1}^{r^*}  \ket{\varphi_i}\bra{\varphi_i}    \rho \Big),
\end{eqnarray}
where the first equality uses the explicit form of $U(\bm{\theta}^*)$.

Recall $\rho=\sum_{i=1}^{r^*} \lambda_i \ket{\psi_i}\bra{\psi_i}$, where $\sum_{i=1}^{2^N} \lambda_i=1$, $0\leq \lambda_i\leq 1$ for $\forall i \in[2^N]$. Combining the explicit form of $\rho$ with the last term in Eqn.~(\ref{eqn:thm1-proof-2}), we obtain 
\begin{equation}\label{eqn:thm1-proof-3}
	\Tr\Big( \sum_{i=1}^{r^*}  \ket{\varphi_i}\bra{\varphi_i}    \rho \Big) = \sum_{i=1}\lambda_i \sum_{j=1}^{r^*}|\langle \varphi_j| \psi_i \rangle |^2.
\end{equation}  
Since $\sum_{i=1}^{r^*} \lambda_i=1$ and $\{\ket{\varphi_i}\}_{i=1}^{2^N}$ ($\{\ket{\psi_i}\}_{i=1}^{2^N}$) is a set of orthonormal basis, the term $\Tr\left( \sum_{i=1}^{r^*}  \ket{\varphi_i}\bra{\varphi_i}    \rho \right)$ equals to one if and only if $\sum_{j=1}^{r^*}|\langle \varphi_j| \psi_i \rangle |^2=1$. Due to the orthonormality of $\{\ket{\varphi_j}\}$ and $\{\ket{\psi_i}\}$,  such an equality can be achieved when $\ket{\psi_i}$ for $\forall i\in[r^*]$ can be linearly spanned by $\{\ket{\varphi_j}\}_{j=1}^{r^*}$, i.e., for $\forall i\in[r^*]$,  
\begin{equation}\label{eqn:opt_U_thm5}
	\ket{\psi_i} =  \ket{\varphi_{j'}},\quad \exists j'\in [r^*] 
\end{equation} 
which ensures $|\langle \varphi_{j'}| \psi_i \rangle |^2 = 1$.  	
	
Combining Eqn.~(\ref{eqn:opt_U_thm1}) and 	Eqn.~(\ref{eqn:opt_U_thm5}),  the optimal quantum encoder yields the form $\sum_{i=1}^{r^*}  \ket{0}^{\otimes (N-K)}\ket{\varpi_i}\bra{\psi_i}  + \sum_{i=r^*+1}^{2^N}   \ket{\phi_i}\bra{\psi_i}$. Since the set of orthonormal basis $\{\ket{\varpi_i}\}$ is non-unique, there are multiple critical points.

We end the proof by indicating that any other form of $U(\bm{\theta})$ differing with Eqn.~(\ref{eqn:opt_U_thm1}) does not promise the optimal result, i.e., $\mathcal{L}(U(\bm{\theta}), \rho)>0$. Specifically, we define such a quantum encoder as
\begin{equation}\label{eqn:opt_U_thm4}
	U(\bm{\theta}) = \sum_{i=1}^{2^N}  \ket{\phi_i'}\bra{\varphi_i'}, 
\end{equation}
 where $\{\ket{\phi_i'}\}$  (or $\{\ket{\varphi'_i}\}$) is a set of orthonormal basis to form the left (or right) singular vectors of $U(\bm{\theta})$. Notably, to guarantee that a subset  orthonormal basis in $\{\ket{\phi_i'}\}$ can not be spanned by $\{\bra{0^{\otimes (N-K)}, \varpi_{i}'}\}$, or equivalently, $U(\bm{\theta})\neq U(\bm{\theta}^*)$, we impose a restriction on $U(\bm{\theta})$  such that \begin{equation}\label{eqn:opt_U_thm3}
 	(\ket{0}\bra{0})^{\otimes(N-K)}\otimes \mathbb{I}_K  ) \sum_{j=1}^{2^N}      |\phi_j' \rangle\bra{\varphi_j'} \prec \sum_{j=1}^{2^N}  | 0^{\otimes (N-K)}, \varpi_{j}' \rangle \bra{\varphi_j'}. 
 \end{equation}

The restriction in Eqn.~(\ref{eqn:opt_U_thm3}) ensures the following result, i.e.,  
 \begin{eqnarray}\label{eqn:thm1-proof-4}
 &&	U(\bm{\theta})^{\dagger}(\ket{0}\bra{0})^{\otimes(N-K)}\otimes \mathbb{I}_K  )U(\bm{\theta})  \nonumber\\ 
 = &&  U(\bm{\theta})^{\dagger}(\ket{0}\bra{0})^{\otimes(N-K)}\otimes \mathbb{I}_K  ) (\ket{0}\bra{0})^{\otimes(N-K)}\otimes \mathbb{I}_K  )U(\bm{\theta})  \nonumber\\
\prec && \sum_{j=1}^{2^N}  \ket{\varphi_j'} \bra{ 0^{\otimes (N-K)}, \varpi_{j}' } \sum_{k=1}^{2^N}  | 0^{\otimes (N-K)}, \varpi_{k}' \rangle \bra{\varphi_k'} \nonumber\\
 = && \sum_{j=1}^{r^*}  \ket{\varphi'_j}\bra{\varphi'_j},
 \end{eqnarray}
 where the first equality uses the property of the projector with $\Pi^2=\Pi$, the first inequality employs Eqn.~(\ref{eqn:opt_U_thm3}), and the last equality uses the orthonormality of $\{\ket{\varpi_{k}'}\}$.   
 
Connecting Eqn.~(\ref{eqn:thm1-proof-2}) with Eqn.~(\ref{eqn:thm1-proof-4}), we conclude that when $U(\bm{\theta})$ satisfies the form in Eqn.~(\ref{eqn:opt_U_thm4}), the following relation holds
\begin{eqnarray}
	&&  \Tr\left(\left(\ket{0}\bra{0})^{\otimes(N-K)}\otimes \mathbb{I}_K  \right) U(\bm{\theta})\rho U(\bm{\theta})^{\dagger} \right) \nonumber\\	
	< &&  \Tr\Big( \sum_{i=1}^{r^*}  \ket{\varphi_i'}\bra{\varphi_i'}    \rho \Big) \nonumber\\
	 \leq &&  1.
\end{eqnarray}
Such a relation implies $\mathcal{L}(U(\bm{\theta}),\rho) >0$, according to Eqn.~(\ref{eqn:proof-thm1-1}). 
	
 \end{proof}

\section{Proof of  Lemma \ref{lem:sigma_eigen} }\label{append:proof_lem1}  
\begin{proof}[Proof of Lemma \ref{lem:sigma_eigen}]

Let us first derive the explicit form of the compressed state $\sigma^*$ based on the spectral decomposition of $U(\bm{\theta}^*)$ and $\rho$. Specifically, based on the explicit form of the optimal compressed state $\sigma^*$ in Lemma~\ref{lem:sigma_eigen}, we have
   \begin{eqnarray}
	\sigma^* = && \Tr_{E}\Big(\frac{M_G U(\bm{\theta}^*)\rho U(\bm{\theta}^*)^{\dagger} M_G^{\dagger}}{\Tr(M_G U(\bm{\theta}^*)\rho U(\bm{\theta}^*)^{\dagger})} \Big) \nonumber\\
	= && \Tr_E\Big( \sum_{i=1}^{r^*} \lambda_i \ket{0^{\otimes (N-K)},\varpi_i} \bra{0^{\otimes (N-K)},\varpi_i} \Big) \nonumber\\
	 = && \sum_{i=1}^{r^*} \lambda_i  \ket{\varpi_i}  \bra{\varpi_i},
\end{eqnarray}
where the second equality uses Eqn.~(\ref{eqn:eigen-rho}) and Eqn.~(\ref{eqn:thm1_0}); the third equality employs the property of orthonormal vectors, i.e., $\langle \psi_i|\psi_j \rangle =\bm{\delta}_{ij}$; the last equality comes from the property of partial trace. Recall the spectral decomposition of $\rho$ in Eqn.~(\ref{eqn:eigen-rho}). An immediate observation is $\sigma_*$ and $\rho$ share the same eigenvalues $\{\lambda_i\}_{i=1}^{r^*}$.

We next elucidate how to use $\ket{\varpi_i}$ and $U(\bm{\theta}^*)$ to recover the eigenvectors $\{\ket{\psi_i}\}$ of $\rho$. In particular, given access to   $\ket{\varpi_i}$, we can prepare the input state $\ket{0}^{\otimes N-K}\ket{\varpi_i}$ and operate it with the optimal unitary $U(\bm{\theta}^*)^{\dagger}$, where the generated state is 
\begin{eqnarray}
	&& U(\bm{\theta}^*)^{\dagger}\ket{0}^{\otimes N-K}\ket{\varpi_i} \nonumber \\ 
	= && 
	\Big(\sum_{i=1}^{r^*} \ket{\varphi_i}\bra{0}^{\otimes N-K}\bra{\varpi_i} + \sum_{j=r+1}^{2^N}\ket{\varphi_j}\bra{\phi_j}\Big)\ket{0}^{\otimes N-K}\ket{\varpi_i} \nonumber\\
	 = && \ket{\varphi_i},
\end{eqnarray} 
where the first equality employs the spectral form of $U(\bm{\theta}^*)$ in Eqn.~(\ref{eqn:thm1_0}). As explained in the proof of Theorem \ref{thm_QAE_noiseless} (i.e., Eqn.~(\ref{eqn:opt_U_thm5})), the eigenvectors $\{\ket{\psi_i}\}$ are linearly spanned by $\{\ket{\varphi_i}\}$. In other words, the prepared state $\ket{\varphi_i}$ can recover the spectral information of $\rho$.

\end{proof}

\section{Proof of Theorem \ref{thm:fide_bound} }\label{append:thm2}

Two results used in the proof of Theorem \ref{thm:fide_bound} are as follows.
\begin{lem}[Lemma 1, \cite{cerezo2020variational}]\label{lem:est_fide}
	For any positive semi-definite operators $A$, $B$, and $C$, where $C$ is normalized, i.e., $\Tr(C) = 1$, the difference of the fidelity $F(A,C)$ and the fidelity $F(B,C)$ yields 
	\begin{eqnarray}
		&& \left|\FF(A, C)-\FF(B,C)\right| \nonumber\\
 	 = &&  \|\sqrt{A}\sqrt{C}\|_1 - \|\sqrt{B}\sqrt{C}\|_1  \nonumber\\ 
 	 \leq &&  \sqrt{\Tr\Big((\sqrt{A} - \sqrt{B})^2 \Big)} \nonumber\\
 	 = && \sqrt{ 2\Big(1 - \Tr\Big(\sqrt{A} \sqrt{B} \Big) \Big)}.
	\end{eqnarray}
\end{lem}

We are now ready to leverage Lemma \ref{lem:est_fide}  to prove Theorem \ref{thm:fide_bound}. 
\begin{proof}[Proof of Theorem \ref{thm:fide_bound}]

Here we evaluate the discrepancy between $\hFide(\rho,\kappa)$ and $\FF(\rho,\kappa)$. In particular, the difference $ |\hFide(\rho,\kappa) - \FF(\rho,\kappa)|$ yields 

\begin{eqnarray}\label{eqn:error_TFB_thm2_1}
 &&  \left|\hFide(\rho,\kappa) - \FF(\rho,\kappa) \right| \nonumber\\
	\leq && \sqrt{ 2\Big(1 - \Tr\Big(\sqrt{\hat{\rho}} \sqrt{\rho} \Big) \Big)} \nonumber\\
	\leq && \sqrt{   \Tr(|\hat{\rho} - \rho|) }, 
\end{eqnarray}
where the first inequality uses the result of Lemma \ref{lem:est_fide}, and the second  inequality exploits  $ 1 - \Tr(\sqrt{\hat{\rho}} \sqrt{\rho} ) \leq \frac{1}{2} \Tr(|\hat{\rho} - \rho|)$  supported by \cite{audenaert2012comparisons}.  
 
Then, according to the Fuchs-van de Graaf's inequality \cite{nielsen2010quantum}, we have
\begin{equation}\label{eqn:varsigna-fidelity}
	\frac{1}{2}\Tr(|\rho - \hat{\rho}|) \leq \sqrt{1 - \FF^2(\hat{\rho}, \rho)} = \varsigma.
\end{equation}

In conjunction with the above results, we obtain that the fidelity $\FF(\rho,\kappa)$ is bounded by the estimated fidelity with an additive error term, i.e., 
\begin{equation}\label{eqn:error_TFB_thm2_lowb}
	\hFide(\rho,\kappa) - \sqrt{\varsigma} \leq  \FF(\rho, \kappa)\leq  \hFide(\rho,\kappa) + \sqrt{2\varsigma}.
\end{equation}

\end{proof}  

\section{On the relation of the training loss $\delta$ and the quantity $\varsigma$}\label{append:train-loss-vs-error-quantity}

In this section, we interpret the training loss of QAE $\mathcal{L}$ in Eqn.~(\ref{eqn:loss_G1}), the trace distance between the target state $\rho$ and the reconstructed state of QAE $\hat{\rho}$ (i.e., $\|\rho - \hat{\rho}\|_1/2$),  and the error quantity $\varsigma$ in Eqn.~(\ref{eqn:est-fide-error-quantity}) yields the following relation
\begin{equation}\label{eqn:relation-loss-varsigma}
	1 - \sqrt{1 - \delta} \leq \frac{1}{2}\|\rho - \hat{\rho}\|_1 \leq \varsigma.
\end{equation} 
This relation can be attained by separately showing (i) $\|\rho - \hat{\rho} \|_1/2\leq \varsigma$; (ii) $1 - \sqrt{1 - \delta} \leq \|\rho - \hat{\rho} \|_1/2$.

\noindent\textit{(i) $\|\rho - \hat{\rho} \|_1/2 \leq \varsigma$}. This relation has been proved in Eqn.~(\ref{eqn:varsigna-fidelity}).

\noindent\textit{(ii) $\delta \leq \|\rho - \hat{\rho} \|_1$}. To achieve this relation, we only need to illustrate that the training loss of QAE is upper bounded by the infidelity of the target state and the reconstructed state, i.e., 
\begin{equation}\label{eqn:train-loss-infide}
	1 - \sqrt{1 - \delta} \leq 1- \FF(\rho, \hat{\rho})
\end{equation}
Then, according to the Fuchs-van de Graaf's inequality \cite{nielsen2010quantum} with
\begin{equation}
	1- \FF(\rho, \hat{\rho}) \leq \frac{1}{2}\|\rho - \hat{\rho}\|_1,
\end{equation}
it immediately obtains $\delta \leq \|\rho - \hat{\rho}\|/2$.

The inequality in Eqn.~(\ref{eqn:train-loss-infide}) is hold due to the monotonicity of fidelity. Note that the training loss of QAE can be reformulated as
\begin{eqnarray}
	&& \mathcal{L}(U(\bm{\theta}),\rho) \nonumber\\
=  && \Tr\left((\mathbb{I}_N-(\ket{0}\bra{0})^{\otimes(N-K)}\otimes \mathbb{I}_K)U(\bm{\theta})\rho U(\bm{\theta})^{\dagger} \right)\nonumber\\
= &&  1 - \Tr\left(\ket{0}\bra{0})^{\otimes(N-K)} \Tr_{\bar{E}}(U(\bm{\theta})\rho U(\bm{\theta})^{\dagger})\right) \nonumber\\
= &&   1 - \FF^2((\ket{0}\bra{0})^{\otimes(N-K)}, \Tr_{\bar{E}}(U(\bm{\theta})\rho U(\bm{\theta})^{\dagger})) 
\end{eqnarray}
where the second equality uses the property of partial trace and the operation $\Tr_{\bar{E}}(\cdot)$ follows Eqn.~(\ref{eqn:sigma_T_1}), which refers to the partial trace on $K$ latent qubits, and the third inequality employs the definition of fidelity when one state is the pure state. In other words, we have
\begin{equation}
	\FF((\ket{0}\bra{0})^{\otimes(N-K)}, \Tr_{\bar{E}}(U(\bm{\theta})\rho U(\bm{\theta})^{\dagger})) = \sqrt{1 - \delta}. 
\end{equation}

Moreover, the fidelity between the reconstructed state and the target state satisfies the following relation, i.e.,
\allowdisplaybreaks
\begin{eqnarray}
	&& 1 -  \FF(\rho, \hat{\rho}) \nonumber\\
= &&  1 - \FF\left(\rho, U(\bm{\theta})^{\dagger}\left((\ket{0}\bra{0})^{\otimes(N-K)}\otimes \Tr_{E}(U(\bm{\theta})\rho U(\bm{\theta})^{\dagger})\right)U(\bm{\theta})\right) \nonumber\\
= &&  1 - \FF\left((\ket{0}\bra{0})^{\otimes(N-K)}\otimes \Tr_{E}(U(\bm{\theta})\rho U(\bm{\theta})^{\dagger}), U(\bm{\theta}) \rho U(\bm{\theta})^{\dagger}\right) \nonumber\\
\geq && \FF((\ket{0}\bra{0})^{\otimes(N-K)}, \Tr_{\bar{E}}(U(\bm{\theta})\rho U(\bm{\theta})^{\dagger})), \nonumber
\end{eqnarray}
where the first equality follows the explicit form of the reconstructed state, the second equality uses the invariant property of fidelity with respect to an isometry, and the inequality comes from the monotonicity of fidelity as $\FF(\rho, \sigma)\leq \FF(\mathcal{N}(\rho), \mathcal{N}(\sigma))$ with $\mathcal{N}$ being a feasible quantum channel.

In conjunction with the above two equations, it is evident that 
\begin{equation}
	1 - \FF(\rho, \hat{\rho}) \geq 1 - \sqrt{1 - \delta}.
\end{equation}

\section{Proof of Lemma \ref{lem:QFI_QAE}}\label{append:sec:lem-QFI}

\begin{proof}[Proof of Lemma \ref{lem:QFI_QAE}]
	The results shown in Lemma \ref{lem:QFI_QAE} can be easily obtained by leveraging Theorem \ref{thm:fide_bound}. Recall the explicit form of the estimated QFI, i.e.,
	\begin{equation}
		\hat{\mathcal{I}}_{\tau}(\vartheta, \rho_{\vartheta}(\bm{\gamma})) = 8\frac{1 - \hFide(\rho_{\vartheta}(\bm{\gamma}), \rho_{\vartheta+\tau}(\bm{\gamma}))}{\tau^2}.
	\end{equation}
Combining the above equation with Eqn.~(\ref{eqn:error_TFB_thm2_lowb}), we have
\begin{equation}
\begin{cases}
	\hat{\mathcal{I}}_{\tau}(\vartheta, \rho_{\vartheta}(\bm{\gamma})) \leq    \mathcal{I}_{\tau}(\vartheta, \rho_{\vartheta}(\bm{\gamma})) + \frac{8\sqrt{2\varsigma}}{\tau^2} \\
	\hat{\mathcal{I}}_{\tau}(\vartheta, \rho_{\vartheta}(\bm{\gamma})) \geq  \mathcal{I}_{\tau}(\vartheta, \rho_{\vartheta}(\bm{\gamma}))- 	\frac{8\sqrt{2\varsigma}}{\tau^2}. 
\end{cases}
\end{equation}	
After simplification, we have 
\begin{equation}
\hat{\mathcal{I}}_{\tau}(\vartheta, \rho_{\vartheta}(\bm{\gamma})) -  \frac{8\sqrt{2\varsigma}}{\tau^2} \leq \mathcal{I}_{\tau}(\vartheta, \rho_{\vartheta}(\bm{\gamma})) \leq  \hat{\mathcal{I}}_{\tau}(\vartheta, \rho_{\vartheta}(\bm{\gamma})) + \frac{8\sqrt{2\varsigma}}{\tau^2}.\nonumber  
\end{equation} 
	 
\end{proof}

\section{Proof of Lemma \ref{thm:QAE-Gibbs}}\label{append:subsec:gibbs-lem1}

Before presenting the proof of  \ref{thm:QAE-Gibbs}, let us first present a necessary lemma, which connects the relative entropy with the fidelity. 
\begin{lem}[Lemma 2, \cite{wang2020variational}]\label{lem:VGS1}
	Given quantum states $\rho$ and $\kappa$ and a constant $\epsilon_2$, suppose that the relative entropy $S(\rho||\kappa)$ is less than $\epsilon_2$. Then the fidelity between $\rho$ and $\kappa$ satisfies
	\begin{equation}
		\FF(\rho, \kappa)\geq 1 - \sqrt{2S(\rho||\kappa)}\geq 1 - \sqrt{2\epsilon_2}.
	\end{equation}
\end{lem}
 
We next present the proof of Lemma \ref{thm:QAE-Gibbs}.
\begin{proof}[Proof of Lemma \ref{thm:QAE-Gibbs}]
	Supported by Lemma \ref{lem:VGS1}, we have 
	\begin{equation}\label{eqn:thm3_gibbs_0}
		\FF(\rho_G, \rho(\hat{\bm{\gamma}}))\geq 1 - \sqrt{2S(\rho(\hat{\bm{\gamma}})||\rho_G)}.  
	\end{equation} 
In other words, the lower bound of $\FF(\rho_G, \rho(\hat{\bm{\gamma}}))$ can be achieved by deriving the upper bound of $S(\rho(\hat{\bm{\gamma}})||\rho_G)$. Specifically, we have
\begin{eqnarray}\label{eqn:thm3_gibbs_1}
	&& S(\rho(\hat{\bm{\gamma}})||\rho_G) \nonumber\\
	 = && \beta|\mathcal{F}(\rho(\hat{\bm{\gamma}})) - \mathcal{F}(\rho_G)| \nonumber\\
	= &&  \beta\left|\mathcal{F}(\rho(\hat{\bm{\gamma}})) - \hat{\mathcal{F}}(\rho(\hat{\bm{\gamma}})) + \hat{\mathcal{F}}(\rho(\hat{\bm{\gamma}}))   - \mathcal{F}(\rho_G)\right| \nonumber\\
	 \leq &&  \beta\left| \mathcal{F}(\rho(\hat{\bm{\gamma}})) -  \hat{\mathcal{F}}(\rho(\hat{\bm{\gamma}}))  \right| + \beta\left| \hat{\mathcal{F}}(\rho(\hat{\bm{\gamma}}))  -  \mathcal{F}(\rho_G)  \right| \nonumber\\
	 \leq &&  \beta\left| \mathcal{F}(\rho(\hat{\bm{\gamma}})) -  \hat{\mathcal{F}}(\rho(\hat{\bm{\gamma}}))  \right| + \beta \epsilon_1, 
\end{eqnarray}  	     
where the first equality stems from $\mathcal{F}(\rho)=\mathcal{F}(\sigma)+\beta^{-1}S(\rho||\sigma)$ \cite{wilming2017axiomatic}, the first inequality uses the triangle inequality, and the second inequality uses the condition $|\hat{\mathcal{F}}_R(\rho(\hat{\bm{\gamma}})) - \mathcal{F}(\rho_G) |\leq \epsilon_1$.

In conjunction with Eqns.~(\ref{eqn:thm3_gibbs_0}) and (\ref{eqn:thm3_gibbs_1}), the quantification of the lower bound of $\FF(\rho_G, \rho(\hat{\bm{\gamma}}))$ amounts to upper bounding the difference between $\mathcal{F}(\rho(\hat{\bm{\gamma}}))$ and  $\hat{\mathcal{F}}(\rho(\hat{\bm{\gamma}}))$. Following the explicit form of $\mathcal{F}$, the explicit form of the difference between $\mathcal{F}_R(\rho(\bm{\gamma}))$ and  $	\hat{\mathcal{F}}_R(\rho(\bm{\gamma}))$ yields
\begin{eqnarray}\label{eqn:lem7_eqn1}
	&& \left|\mathcal{F}(\rho(\bm{\gamma})) - 	\hat{\mathcal{F}}(\rho(\bm{\gamma})) \right|\nonumber\\
= &&  |S(\rho)- S(\hat{\rho})| \nonumber\\
\leq && \frac{1}{2}\|\rho - \hat{\rho}\|_1 +  \frac{1}{2}\|\rho - \hat{\rho}\|_1 \log(2^N-1) + \frac{1}{2}\|\rho - \hat{\rho}\|_1 \log \frac{2}{\|\rho - \hat{\rho}\|_1} \nonumber\\
\leq && \varsigma  + \varsigma \log(2^N-1) + \varsigma\log \frac{1}{\delta}, 
\end{eqnarray}
where the first equality comes from the fact that the difference between $\mathcal{F}(\rho(\bm{\gamma}))$ and $\hat{\mathcal{F}}(\rho(\bm{\gamma}))$ originates from the estimation error of Von-Neumann entropy, the second inequality leverages the result of Ref.~\cite{reeb2015tight} to upper bound the entropy difference by the trace distance, and the last inequality adopts the relation $1 - \sqrt{1- \delta}\leq \frac{1}{2}\|\rho - \hat{\rho}\|_1\leq \sqrt{1 - \FF^2(\rho, \hat{\rho})}=\varsigma$ as explained in Appendix~\ref{append:train-loss-vs-error-quantity}.

Taken together, we achieve 
\begin{equation}\label{eqn:lem_QAE_Gibbs_4}
	\FF(\rho_G, \rho(\hat{\bm{\gamma}})) \geq  1 - \sqrt{2\varsigma \left(1 + \log \frac{2^N-1}{ 1 - \sqrt{1- \delta}} \right) + 2\beta \epsilon_1 }.
\end{equation}

\end{proof}   
 
\end{document}